\documentclass[11pt, a4paper]{article}
\usepackage{amsmath,amsthm}
\usepackage{amssymb}
\usepackage[dvips]{graphics}
\usepackage[dvips]{graphicx}
\usepackage{soul}
\usepackage{mathtools}

\author{Francesco Giglio}

\theoremstyle{definition}
\newtheorem{theorem}{Theorem}[section]
\newtheorem{lemma}{Lemma}[section]
\newtheorem{definition}{Definition}[section]
\newtheorem{corollary}{Corollary}[section]
\newtheorem{proposition}{Proposition}[definition]

\numberwithin{equation}{section}
\usepackage{color}

\usepackage{amsmath}
\usepackage{amsfonts}
\usepackage{graphicx}
\usepackage{bm}
\usepackage{amssymb}
\usepackage{mathrsfs}
\usepackage{pstricks,pst-node}
\usepackage{amsxtra}
\usepackage{bbm}
\usepackage{paralist}
\usepackage{color}

\usepackage[symbol]{footmisc}

\newcommand{\diag}{\textup{diag}}

\DeclareMathOperator{\Tr}{Tr}

\DeclareMathAlphabet{\mathbfi}{OT1}{cmr}{bx}{it}
\DeclareMathAlphabet{\mathpzc}{OT1}{pzc}{m}{it}

\newcommand{\benumerate}{\begin{enumerate}}
\newcommand{\eenumerate}{\end{enumerate}}

\newcommand{\bitemize}{\begin{itemize}}

\newcommand{\eitemize}{\end{itemize}}

\newcommand{\der}[2]{\frac{\partial #1}{\partial #2}}

\newcommand{\dersec}[2]{\frac{\partial^{2} #1}{\partial #2^{2}}}

\newcommand{\dermixd}[3]{\frac{\partial^{2} #1}{\partial #2 ~\partial #3}}

\newcommand{\av}[1]{\langle #1 \rangle}

\begin{document}

\title{
Complete integrability and equilibrium thermodynamics of biaxial nematic systems with discrete orientational degrees of freedom}

\author{Giovanni De Matteis$^{\;a,b)}$, Francesco Giglio$^{\;c)}$, Antonio Moro$^{\;d)}$}

\date{}

\maketitle

\begin{center}{\small $^{a)}$Dipartimento di Matematica e Fisica, Universit\a`a del Salento, Lecce, Italy \\
		\vspace{0.1cm}
		$^{b)}$I.N.F.N. Sezione di Lecce,  Lecce, Italy \\
\vspace{0.1cm}
{{\small $^{c)}$ School of Mathematics and Statistics, University of Glasgow, Glasgow, UK} 
\vspace{0.1cm}}\footnote[1]{Corresponding Author: francesco.giglio@glasgow.ac.uk}\\
{\small $^{d)}$ Department of Mathematics, Physics and Electrical Engineering, Northumbria University Newcastle, Newcastle upon Tyne, UK}}
\end{center}

    \vspace{-7mm}
\begin{center}

\end{center}

\begin{abstract}
We study a  discrete version of a biaxial nematic liquid crystal model with external fields via an approach based on the solution of differential identities for the partition function. In the thermodynamic limit, we derive the free energy of the model and the associated  closed set of equations of state involving four order parameters, proving the integrability and exact solvability of the model.  The equations of state are specified via a suitable representation of the orientational order parameters, which imply   two-order parameter reductions in the absence of external fields.  A detailed exact analysis of the equations of state reveal a rich phase diagram where isotropic versus uniaxial versus biaxial phase transitions are explicitly described,  including the existence of triple  and tricritical points.  Results on the discrete models are qualitatively consistent with their continuum analog. This observation suggests that, in more general settings, discrete models may be used to capture and describe phenomena that also occur in the continuum for which exact equations of state in closed form are not available.

\vspace{.4cm}

\end{abstract}


\section{Introduction}
Mean-field models in Statistical Mechanics and Thermodynamics are a powerful tool to explore general qualitative properties of thermodynamic systems that, otherwise, would not be analytically treatable. Both conceptual and historical importance of mean-field models is testified by the celebrated van der Waals and Curie-Weiss models~\cite{Stanley}, complemented by Maxwell's equal areas rule (see e.g. \cite{Callen}) which provided the first qualitative description of the mechanisms for the occurrence of phase transitions in fluids and magnetic systems. It is also well established that, in order to obtain accurate quantitative predictions, mean-field models need to be replaced by models with finite range interactions which are generally more challenging, and solvable cases require the use of sophisticated techniques as, for example, the transfer matrix and the renormalisation group, see e.g.~\cite{Parisi}.  \\
Spin models are the archetypal example of  models aimed at describing the macroscopic and collective behaviour of systems made up of components with internal degrees of freedom (in the simple case the spin $\sigma =\pm 1 $) with pairwise (and also higher order) interactions. Such models, although originally introduced in condensed matter physics to explain magnetic properties of materials are, however, of universal importance, as testified by applications in other disciplines such as Biology, Economics, Social Sciences, see e.g.~\cite{BarraContucciGallo,BarraContucciVernia,MoroScirep} and references therein. It is also worth noting that a resurgence of interest, in the last decade, for spin-like mean-field models is due to the studies concerning their deployment for information processing, classification, memory retrieval and, more generally, machine learning purposes \cite{agliari2022}. These studies, originally inspired by the pioneering work of Hopfield~\cite{Hopfield}, led to the definition of models for neural networks, such as the Boltzmann machines and their variations, based on spin glasses and statistical inference algorithms for training and learning~\cite{Contucci}. The key idea in this context is that spin particles sit at a node of a graph and possess internal degrees of freedom, i.e. their spin values are interpreted as node states of the neural network associated to the graph. The spin-spin interaction constant corresponds to the weight associated to the links on the network.\\
In this paper, we consider a biaxial version of the discrete Maier-Saupe model for nematic liquid crystals (LCs) as studied in \cite{dgm nematics}, whose structure  resembles a multi-partite spin model with  spin components subject to suitable constraints.
The model consists of a system of particles endowed with an internal assigned geometry and symmetries with only orientational degrees of freedom. Not surprisingly, the exact analytical description of their macroscopic thermodynamic behaviour, phase transitions and emergent properties is, in general, not available and therefore alternative approaches and approximation techniques need to be adopted~\cite{SonnetVirga}. Numerical simulations \cite{Zannoni}, Landau's expansion of the free energy~ \cite{Gramsbergen,Mukherjee2009},  group representation and bifurcation theory \cite{chillingworth} are approaches that allow to explore, at least locally, i.e. in the neighbourhood of specified values for the thermodynamic parameters, the possible occurrence of criticalities and phase transitions, and estimate relevant thermodynamic quantities such as orientational order, specific heat, critical exponents. Mean-field models are effective in providing insights that complement and support the aforementioned methodologies and all together help achieve accurate qualitative description and predictions on key properties of LCs including those that are paramount for technological applications \cite{Sluckin}.\\
From a physical viewpoint, in the last few decades, the biaxial nematic liquid crystal phase has been the object of much intense study.
The story of this phase has its roots back to 1970 \cite{freiser1970}, when the theoretical physicist Marvin Freiser noted that rather than possessing a rod-like shape (i. e. $D_{\infty h}$ symmetry), as usually assumed, most thermotropic, mesogenic molecules were in fact closer to being board-like,
thus intrinsically biaxial (i. e. endowed with $D_{2h}$ symmetry).
Usually, they produce uniaxial nematic phases as a consequence of the rotational disorder
around the long molecular axis, which eventually yields the definition of a single macroscopic director.
This rotational disorder can be overcome by molecular mutual interactions favoring the molecules to align parallel to one another, thus leading to a thermotropic biaxial nematic phase at
sufficiently low temperatures.
Accordingly, Freiser understood that mesogens should be expected to exhibit a biaxial nematic phase, in addition to the usual uniaxial one.
The prediction of a second nematic phase possessing novel properties and promising potential applications, stimulated considerable interest, as well as not little debate.
In fact, on the experimental side,  stable biaxial phases have been observed in lyotropic systems since the pioneering work of Yu and Saupe \cite{YuSaupe}. In contrast, the experimental proof in favour of their existence
in thermotropic systems has been subject of scrutiny and criticism in \cite{luckhurst2001, berardi1995, berardi2000}. 
In the period 1986 to 2003, the matter remained controversial with no widely accepted results \cite{luckhurst2001, praefcke2001, praefcke2002}.
However, since 2004, clearer experimental evidence was provided for a few classes of compounds, such as polar bent-core or V-shaped
molecules \cite{gortz, senyuk, lehmann}, and organosiloxane tetrapodes or their counterparts with a germanium core \cite{merkel2004, cruz2005, neupane2006, mehl2008, cruz2008}. These compounds have been investigated by several techniques which led to measurements of biaxial order parameters~\cite{southern}.
According to these experimental results, an alternative picture of biaxial nematic order has emerged \cite{berardi2008, tschierske2010, vanakaras2008, peroukidis2009, francescangeli2011}, based on the idea of biaxial domains reoriented by
surface anchoring or external fields. Other researchers \cite{francescangeli2010, samulski2010, zhang2012} have also pointed out that the biaxial nematic order is related
to the onset of smectic fluctuations. Moreover, it has also been remarked that biaxial nematics may be formed from molecules possessing
a lower symmetry than the usually assumed $D_{2h}$ one, as for istance the $C_{2h}$ symmetry \cite{ peroukidis2009,lehmann2011, karahaliou2009, gorkunov2010, osipov2010, dong2010}.
In addition, quite recently \cite{osipov2012, ghoshal2012}, low symmetry interaction models have been addressed, involving dipolar contribution,
so as to describe polar bent-core molecules.
The study of biaxial nematics is not only of theoretical origin, it is also connected with their potential technological applications
in displays \cite{berardi2008, luckhurstnature2004, luckhurstbritish2005, lee2007, berardijcp2008, nagaraj2010}: orientation of the secondary director in response to external perturbations is expected to be significantly faster than  the primary one \cite{tschierske2010, lehmann2011}.
Biaxial nematic phases have also been produced in colloidal suspensions of inorganic compounds \cite{pelletier1999, petukhov2009, petukhov2010}.
More recently, in  \cite{smalyuhk2018},  Smalyuhk et al. have considered a hybrid molecular-colloidal
soft-matter system with orthorhombic biaxial orientational order and fluidity. This molecular-colloidal complex fluid is made up of only uniaxial rod-like building blocks.
In contrast, this complex fluid exhibits a surprising self-assembly into a biaxial nematic liquid crystal with the $D_{2h}$
point group symmetry. Finally, let us mention that very recently, the emergence of biaxial order upon mechanical strain has been proved experimentally in a nematic liquid crystal elastomer, the first synthetic auxetic material at a molecular level \cite{raistrick2022}. By measuring the order parameters during deformation, the deviation from Maier-Saupe theory was detected for the uniaxial order parameters and the biaxial order parameters were deduced, suggesting the occurrence of biaxiality in the initially uniaxial system.\\
On the theoretical side, after Freiser's first prediction \cite{freiser1970}, investigations were actively carried on along different approaches
such as molecular-field or Landau theories, and later on by computer simulations.
By the end of the past century, this collection of theoretical methodologies has shown that single-component models consisting of molecules possessing
$D_{2h}$ symmetry, and interacting via various continuous or hard-core potentials, are capable of producing a biaxial nematic phase under appropriate
thermodynamic conditions \cite{berardi2008, romano2004, pisapv08, preeti2011}.
Theoretical studies usually predict a low-temperature biaxial phase, undergoing a transition to the uniaxial one, which, in turn, finally turns into
the isotropic phase.
In some cases, the transition takes place directly from the biaxial nematic to the isotropic phase.
In the former cases, the ratio between the two transition temperatures (biaxial-to-uniaxial and uniaxial-to-isotropic) often turns out to be rather
small in comparison with experimentally known stability ranges of the nematic phase.
Both the isotropic-to-biaxial and uniaxial-to-biaxial phase transitions can be either first- or second-order, and, accordingly, the phase diagram exhibits
{\emph{triple}} and {\emph{tricritical points}}.
However, in the low temperature range, other phases, such as smectic or solid ones, may become more likely to occur.
On the other hand, most theoretical frameworks only allow for isotropic and nematic phases \cite{tschierske2010}, being the positional order
not accounted for.
Over the years, a rather simple, continuous, biaxial mesogenic pair interaction model has been proposed and investigated by several authors
and via several types of techniques. In the literature, this model is known as the {\emph{generalised Straley interaction}} \cite{gdmromanobisipre2013c} and
it finds its roots in the celebrated Maier-Saupe model for interacting uniaxial nematic molecules \cite{maiersaupe1958,maiersaupe1959, straley1974}.
Actually, over the last two decades, several properties of this model have emerged, such as possible simplifications, additional symmetries and versatility in
applications. 
More precisely, in 2003, new experimental findings on biaxial nematics boosted a renewed theoretical interest by some authors \cite{svd2003, gdmv2005, gdmromanov2005, universal2006, cmt2007, romano2004333}. 
More precisely, the generalised Straley pair potential model was studied by mean-field, as well as Monte Carlo simulation in the simple-cubic lattice-model version
and, correspondingly, the effects produced on the resulting macroscopic behaviour were analysed \cite{romano2004, pisapv08, gdmromanov2005,universal2006, romano2004333, gartlandv2010}.
Moreover, motivated by the new experimental facts, the single-tensor Landau-de Gennes theory of biaxial nematics has been
carefully revisited and a double-tensor Landau theory was put forward and studied \cite{longalandau, gdmcmt2008landau}. The hidden link between mean-field and Landau-de Gennes-type treatments
has also been studied  \cite{sluckin2012,turzi2012,dispersion_london3}.
The Straley potential model involves three independent parameters, and the aforementioned studies have shown that the model is rather
versatile and capable of producing both biaxial and purely uniaxial order.
In addition, the effect of strong antinematic terms, i.e. terms  promoting misalignment, in the pair potential onto the resulting orientational order has been investigated \cite{pisapv08, gdmromano2009}.
As shown in \cite{gdmromano2009, bisigdmromano2012}, these antinematic terms in the Straley model may destroy biaxiality, producing only uniaxial orientational order, and in some cases show evidence of  the existence of
a continuous ordering transition, in contrast with the discontinuous phase transition predicted by the simple Maier-Saupe model.
Moreover, in \cite{gdmromano2011p}, the Straley potential only contains antinematic terms, and it is found to produce biaxial order via a mechanism of order by disorder.
In \cite{gdmromanobisipre2013c} the authors investigated the effect of two predominant antinematic couplings of equal strength perturbed by a comparatively weaker
calamitic one. The resulting phases are a pure calamitic uniaxial phase, accompanied by an intermediate antinematic uniaxial phase.\\
In this work, we consider a {\emph{discrete}} version of the celebrated Maier-Saupe model for nematic LCs as the one considered in \cite{dgm nematics} and study its {\it biaxial} generalisation, i.e.  Straley model, further extended to account for the effects of external fields. More specifically, molecules are assumed to be rigid cuboids, with two individual orientational degrees of freedom associated to two of the three principal axes of inertia, as the position of the third axis is automatically determined.    It is also assumed that homologous principal axes of inertia interact pairwise for any pairs of molecules in the system. This assumption specifically characterises the mean-field models, where indeed any pair of molecules equally interact independently of their distance, and therefore positional degrees of freedom are not relevant. A further assumption is that orientational degrees of freedom are discrete, namely principal axes can only be parallel to the directions of a pre-defined Cartesian reference frame. The discretisation of orientational degrees of freedom for nematic liquid crystal models was firstly introduced by  Zwanzig in \cite{Zwanzig} and successfully employed in various works, including recent papers \cite{dgm nematics,Nascimento}. Although this assumption may seem to be at a glance restrictive, it captures, as observed in \cite{dgm nematics}, with strikingly accuracy, properties of the continuum model. We show, via explicit examples, that the predictions obtained under specific symmetry reductions are consistent with the ones present in the literature for the corresponding continuum models.  \\
It is also important to note that, although, on one hand, the above assumptions restrict the model and allow to derive explicit global equations of the thermodynamic order parameters, on the other hand, the model is more general than its continuum analogues and the methodology adopted naturally incorporates external fields interacting with each orientational degree of freedom.  Therefore, to the best of our knowledge,  we provide the first theoretical study on the equilibrium statistical mechanics of  a molecular field theory for biaxial liquid crystals subject to external fields. \\
To solve the model, we show that the partition function $Z_{N}$ of the $N-$molecules reduced Straley biaxial model with external fields satisfies a remarkable differential identity as a function of the temperature and coupling constants. Using suitably re-scaled independent variables, the differential identity for the partition function of the finite size model is equivalent, up to a linear change of variables, to the heat equation. The required solution is therefore obtained by solving a linear equation with a specific initial condition that is fixed by the value of the partition function for the non interacting model the solution of which is straightforward. The properties of the system in the thermodynamic regime are obtained by studying the behaviour of the free energy
\[
{\cal F}_{N} := \frac{1}{N} \log Z_{N}
\]
in the limit as $N \to \infty$, which corresponds to the {\it semi-classical}, or low diffusion, limit, of the heat equation, via a suitable asymptotic expansion of the free energy in powers of $N^{-1}$. At the leading order, the problem is solved via a Hamilton-Jacobi equation, which can be explicitly integrated and the solution is given in terms of the orientational order parameters from which the equations of state follow as a stationary point for the free energy functional. We study in detail the solution of the Hamilton-Jacobi equation and, specifically, a related system of quasilinear PDEs for the orientational order parameters.\\
The methodology based on differential identities is applied for the first time in the context of a molecular theory for biaxial nematics described by two tensor fields. We derive explicit expressions for state functions when a finite number of molecules $N$ is considered, as well as a novel system of equations of state, which include the interaction with external fields.   We rigorously classify all admissible  reductions in the absence of external fields, revealing  a rich singularity structure describing  transitions between isotropic, uniaxial and biaxial phases. A comparison  with results available in the literature shows that our findings  are consistent with those  obtained with different methods and techniques. \\
As pointed out in a number of papers~\cite{BrankovZagrebnov,ChoquardWagner,GenoveseBarra,DeNittis,AntonioAnnals,BarraMoro,Landolfi,dgm nematics,lmPotts,gvsym,bmps Pstar}  the nature of the PDEs derived for the orientational order parameters suggests a natural interpretation of the singularities as classical shocks propagating in the space of thermodynamic variables. This allows to explain and, qualitatively, predict some features of the phase diagram based on the general properties of shock waves, as for example the occurrence of {\emph{tricritical}} points as a collision and merging mechanism of two shock waves. This example demonstrates how such an  interpretation is at the same time intriguing and of practical use. \\
The paper is organised as follows.
In Section \ref{sec:model} we introduce the physical model under study, we derive differential identities for the statistical partition function and we provide exact solutions for the model in the finite-size regime. In Section \ref{sec:thermodynamic limit and eos}, we perform the thermodynamic limit and derive exact equations of state for the full model.  Two-parameter reductions are also obtained  in the cases of i) zero fields and ii) non-zero fields under special constraints. In Section~\ref{sec:order parameters} we present the phase diagram of the model in  absence of external fields, and discuss   criticality and behavior of the corresponding order parameters. Section~\ref{sec:conclusion} is devoted to concluding remarks. 
    

\section{The discrete $\lambda$-model for  biaxial nematics}
\label{sec:model}
Let us consider a system of $N$ interacting Liquid Crystals molecules with $D_{2h}$ symmetry, whose molecular directors $\vec{m}$, $\vec{e}$ and $\vec{e}_{\perp}$ are mutually orthogonal unit vectors parallel to their principal axes. The orientational state of a given molecule is identified by the directions of its molecular axes. Introducing the tensors (see e.g. \cite{rPV01}) 
\begin{equation}
{\bf q} = \vec{m} \otimes \vec{m} - \frac{1}{3} { \bf I} \qquad ,\qquad {\bf b} = \vec{e} \otimes \vec{e} - \vec{e}_{\perp} \otimes \vec{e}_{\perp}  
\end{equation}
where ${\bf I}$ is the $3 \times 3$ identity matrix, we consider the Hamiltonian of the form 
\begin{equation}
\label{eq:H0}
H_{0} =  -\frac{ \mu}{2N} \sum_{i,j} \left({\bf q}_{i}  \cdot {\bf q}_{j} + \lambda \, {\bf b}_{i}  \cdot {\bf b}_{j} \right)\,,
\end{equation}
where  ${\bf q}_i$ and ${\bf b}_i$ specify the orientational  state of the $i-$th molecule and the scalar product is  $ {\bf a} \cdot {\bf b} := \Tr{({\bf a b})}$, where $\Tr{}$ is the trace operator. Summation indices $i$ and $j$ run from $1$ to $N$,  $\mu$ is  the  non-negative mean-field coupling constant and $\lambda$ is a parameter weighing the degree of biaxiality.
	 In the present paper, we assume $\lambda \in \left[0,1\right]$.
	In this range, the ground state for two interacting molecules corresponds to parallel homologous axes, that is $\bm{e}_i$ tend to line up with $\bm{e}_j$, $\bm{m}_i$ with $\bm{m}_j$
	and $\bm{e}_{\perp, i}$ with $\bm{e}_{\perp,j}$. 
	When $\lambda=0$, the above Hamiltonian reduces to the classical Maier-Saupe model.  
	The specific choice $\lambda=1/3$ corresponds to the MMM model for liquid crystals with equally nematic interaction among corresponding molecular axes  \cite{gdmromanov2005}, i.e.
	\begin{equation}
	H_1 =  -\frac{ \mu}{2N} \sum_{i,j} \left({\bf q}_{i}  \cdot {\bf q}_{j} + \frac{1}{3} \, {\bf b}_{i}  \cdot {\bf b}_{j} \right)\,= -\frac{ \mu}{2N} \frac{2}{3}\sum_{i,j}\left[(\bm{m}_i  \cdot \bm{m}_j)^2 + (\bm{e}_{i} \cdot \bm {e}_{j})^2+
	(\bm{e}_{\perp , i}  \cdot \bm{e}_{\perp , j})^2-\frac{1}{2} \right]\,.
	\end{equation}

For convenience, we have included self-interaction terms corresponding to $i=j$. This choice will not affect the result as it corresponds to a shift of the energy reference frame by a constant. 

The Hamiltonian \eqref{eq:H0} corresponds to a   Straley pair-interaction potential  reduced to the case of explicit zero-coupling between  ${\bf q}$ and ${\bf b}$ tensors \cite{straley1974,svd2003}.   Such pair-potential, identifying the so-called  \emph{$\lambda-$model}, has been studied extensively and its associated  phase  diagram,   in the absence of external fields,  has been inferred for specific two-order parameter reductions \cite{svd2003,gdmv2005,gdmromanov2005}. It is worth noticing that, although the two tensors  ${\bf q}$ and ${\bf b}$ are not directly coupled in the Hamiltonian  \eqref{eq:H0},  they are geometrically  related via the constraint ${\bf q}_i \cdot{\bf b}_i=0$, thus implying an implicit microscopic coupling. As a result, the macroscopic behaviour will eventually reflect this hidden coupling through an entropic contribution in the free energy, in addition to other possible coupling terms in the order tensors, as we also show in this work.

Assuming that allowed configurations are such that molecular directors are parallel to the axes of a fixed Cartesian reference frame, the Hamiltonian (\ref{eq:H0}) can be written as follows
\begin{align}
\notag
H_0&=
-\frac{\mu}{2 N} \sum_{i,j} \sum_{l,k\in \{1,2\}} c_{kl} \left( \Lambda_{i}^{l}\Lambda_{j}^{k}+\lambda \, \Lambda_{i}^{l+2}\Lambda_{j}^{k+2} \right)\, ,
\end{align}
where $c_{kl}=1+\delta_{kl}$ for $k,l=1,2$, and $\Lambda_i^l$, with $i=1,\cdots,N$, and $l=1,2,3,4$, parametrise the components of ${\bf q}_{i}$ and ${\bf b}_{i}$ as follows
\begin{equation}
\label{MS6defintionOmega}
{\bf q}_{i} = \diag(\Lambda^{1}_i,\Lambda^{2}_i, -\Lambda^{1}_i-\Lambda^{2}_i) \qquad , \qquad {\bf b}_{i} = \diag(\Lambda^{3}_i,\Lambda^{4}_i, -\Lambda^{3}_i-\Lambda^{4}_i)
\end{equation}
 giving six possible orientational states of each molecule. 
In particular, we have that for the $i-$th molecule, $\Lambda_i=\left(\Lambda_i^1,\Lambda_i^2,\Lambda_i^3,\Lambda_i^4 \right)\in \{ \Lambda^{(1)}, \Lambda^{(2)},\cdots,\Lambda^{(6)} \}$, where 
\begin{subequations}
\begin{align} 
	\label{eq:quadruples1}
	\Lambda^{(1)}&=\left( \frac{2}{3},-\frac{1}{3},0,-1 \right) \quad  \quad \Lambda^{(2)}=\left( \frac{2}{3},-\frac{1}{3},0,1 \right)\qquad \, \, \,  \quad \Lambda^{(3)}=\left( -\frac{1}{3},\frac{2}{3},1,0 \right)  \\
	\label{eq:quadruples2}
	\Lambda^{(4)}&=\left( -\frac{1}{3},\frac{2}{3},-1,0 \right)\quad \quad \Lambda^{(5)}=\left( -\frac{1}{3},-\frac{1}{3},-1,1\right) \quad \quad \Lambda^{(6)}=\left( -\frac{1}{3},-\frac{1}{3},1,-1\right)\,.
\end{align}
\end{subequations}
Upon introducing the quantities $M^l= \sum_{i} \Lambda_{i}^{l}/N$ with $l=1,2,3,4$, the Hamiltonian $H_{0}$ reads as follows
\begin{equation}
\label{ms6hamiltonian}
H_0= -\mu N \left[ (M^{1}) ^{2} + M^1 M^2 + (M^{2})^{2} + \lambda \left( (M^3) ^{2} + M^3 M^4 + (M^{4})^{2} \right) \right].
\end{equation}
We now proceed with modelling the interaction between the liquid crystal and external fields. Consistently with previous studies on  uniaxial \cite{Wojtowicz, Rosenblatt,Frisken,dunmur88,Mukherjee2013} and biaxial nematics \cite{longa_external,dgm nematics}, we assume that the interaction between an individual biaxial liquid crystal molecule and external fields produces a term that is linear in the molecular tensors. 

Let   ${\boldsymbol \epsilon}= \diag \left (\epsilon_1, \epsilon_2,\epsilon_3 \right)$ and ${\boldsymbol \chi}= \diag \left (\chi_1, \chi_2,\chi_3 \right)$ be two tensors associated with a general external field and   let $H_{{\text ex}}$ be the Hamiltonian modelling the interaction between the external field and the liquid crystal molecules. Our assumption implies that  $H_{ {\text ex}}$ is of the form 
\begin{align} 
H_{{\text ex}}&= -  \sum_{i} \left( \boldsymbol{\epsilon} \cdot {\bf q}_{i}+ \boldsymbol{\chi} \cdot {\bf b}_{i} \right) \\
&=-N \left[ (\epsilon_1 -\epsilon_3 )M^1 + (\epsilon_2 -\epsilon_3 ) M^2 +(\chi_1 -\chi_3 )M^3 + (\chi_2 -\chi_3 )M^4 \right] \,.
\end{align}
By introducing the notation $\epsilon_{k3}=\epsilon_k -\epsilon_3$ and $\chi_{k3}=\chi_k -\chi_3$ with $k=1,2$, we can write\footnote{The parameters $\epsilon_{k,j}$ and $\chi_{k,j}$ can be thought as functions of the external field applied and the properties of the material, e.g. the components of the  magnetic and electric susceptibilities (see for instance  \cite{Wojtowicz, Rosenblatt,Mukherjee2013}).}
\begin{equation}
\label{eq:Hex}
H_{{\text ex}} = - N \left( \epsilon_{13} M^1 + \epsilon_{23} M^2+ \chi_{13} M^3 +\chi_{23} M^4 \right)\,.
\end{equation}
Hence, the full Hamiltonian for the mean-field model under study in this work   is $H = H_{0} + H_{{\text ex}}$. The associated partition function for the Gibbs distribution is given by the expression
\[
Z_{N} = \sum_{\{ ({\bf q},{\bf b}) \}} \exp (-\beta H),
\]
where the summation refers to all possible configurations of  $({\bf q}_i,{\bf b}_i)$ and $\beta = 1/T$ with $T$ denoting the absolute temperature. Upon introducing the rescaled coupling constants $t:= \beta \mu$, $x :=  \beta \epsilon_{13}$, $y :=  \beta \epsilon_{23}$, $z :=  \beta \chi_{13}$ and $w :=  \beta \chi_{23}$,  the partition function reads as
\begin{equation}
\label{eq:partition}
Z_{N} = \sum_{\{ ({\bf q},{\bf b}) \}} e^{N \left \{ t \left[(M^1)^2 + M^1 M^2 + (M^2)^2   + \lambda \left( (M^3) ^{2} + M^3 M^4 + (M^{4})^{2} \right) \right] + x M^1 + y M^2+ z M^3 + w M^4 \right \}} ~.
\end{equation}
In the following, similarly to the case of  van der Waals type models \cite{DeNittis,BarraMoro}, spin systems \cite{GuerraProcA} and the  generalisation of the Maier-Saupe model in \cite{dgm nematics}, we look for a differential identity satisfied by the partition function and calculate the associated initial condition. 
We observe that the partition function~(\ref{eq:partition}) satisfies the $(4+1)$-dimensional linear PDE
\begin{equation}
\label{eq:heat}
\der{Z_{N}}{t} = \frac{1}{N} \left [ \dersec{Z_{N}}{x} + \dermixd{Z_{N}}{x}{y} + \dersec{Z_{N}}{y} + \lambda \left(\dersec{Z_{N}}{z} + \dermixd{Z_{N}}{z}{w} + \dersec{Z_{N}}{w} \right) \right].
\end{equation}
Note that, for $\lambda>0$ , equation (\ref{eq:heat}) can be transformed via a linear transformation of the spatial coordinates into the heat equation
\begin{equation}
	\notag
	\der{Z_{N}}{t} = \sigma \left (\dersec{Z_{N}}{x'} + \dersec{Z_{N}}{y'} + \dersec{Z_{N}}{z'} + \dersec{Z_{N}}{w'} \right) \,,
\end{equation}
 where $x'$,$y'$, $z'$, $w'$ denote the new coordinates and $\sigma=1/N$ is the analogue of the heat conductivity. More precisely, the transformation of coordinates is given by  ${\bf{u} '}=P_\lambda \bf{u }$, where 
 \[{\bf{u} '}=(x',y',z',w')^T,  \, \, {\bf{u}}=(x,y,z,w)^T \, \, \text{and} \, \,   P_\lambda=
 \begin{pmatrix}
	2 & -2 & 0 & 0 \\
	2/\sqrt{3} & 2/\sqrt{3} &0& 0 \\
	0 & 0 & 2/\sqrt{\lambda} & -2/\sqrt{\lambda}\\
	0 & 0 & 2/\sqrt{3\lambda} &  2/\sqrt{3\lambda}
	\end{pmatrix} \,.
  \]
The associated initial condition, $Z_{0,N}(x,y,z,w):=Z_{N}(x,y,z,w,t=0)$, corresponds to the value of the partition function of the model for non mutual interacting molecules.
Given that the exponential is linear in the variables $M^1$, $M^2$, $M^3$ and $M^4$, the initial condition can be evaluated by recursion and gives the following formula
\begin{equation}
\label{eq:partition0}
Z_{0,N} = \left( \sum_{i=1}^{6} e^{x\Lambda^{1,i} + y \Lambda^{2,i}+z\Lambda^{3,i} + w \Lambda^{4,i}} \right)^N,
\end{equation} 
where the index $i$ labels the quadruples $\Lambda^{(i)}=\left( \Lambda^{1,i},\Lambda^{2,i},\Lambda^{3,i},\Lambda^{4,i}\right)$ defined in Eqs. (\ref{eq:quadruples1})-(\ref{eq:quadruples2}).

The exact solution to the equation~(\ref{eq:heat}) for a given number of molecules $N$ can be  formally obtained by separation of variables using as a basis the set of exponential functions obtained by expanding the $N-$th power at the r.h.s. of equation (\ref{eq:partition0}). The solution reads as
\begin{equation}
\label{finiteNZ}
Z_N = \sum_{\{\vec{k} \}} B_{\vec{k}} \; A_{\vec{k}}(t;\lambda)  \exp \left( x \, \omega^{1}_{\vec{k}} + y \, \omega^{2}_{\vec{k}} +z \, \omega^{3}_{\vec{k}} + w \, \omega^{4}_{\vec{k}} \right),
\end{equation}
where $\vec{k} = (k_{1},\dots,k_6)$ is a multi-index such that $k_i = 0,\dots,N_i$ with $N_1=N$, $N_i = N_{i-1}-k_{i-1}$ for $i=2,\dots,5$, $k_6 = N -\sum_{i=1}^{5} k_i$, $\omega^{l}_{\vec{k}}  =\sum_{i=1}^{6}\Lambda^{l,i} k_i $, $l=1,2,3,4$ and
\[
B_{\vec{k}} =\prod_{i=1}^6 \binom{N_i}{k_i}, \quad A_{\vec{k}} = \exp \left \{ \frac{t}{N} \left[  \left(\omega^{1}_{\vec{k}} \right)^2  +\omega^{1}_{\vec{k}} \omega^{2}_{\vec{k}} +\left(\omega^{2}_{\vec{k}} \right)^2  +\lambda \left( \left(\omega^{3}_{\vec{k}} \right)^2  +\omega^{3}_{\vec{k}} \omega^{4}_{\vec{k}} +\left(\omega^{4}_{\vec{k}} \right)^2 \right)  \right]   \right \}.
\]
Let us define the scalar order parameters  $m_{N}^1$, $m_{N}^{2}$, $m_{N}^3$ and $m_{N}^{4}$ as the expectation values of, respectively,  $M^1$, $M^2$, $M^3$ and $M^4$, i.e.
\begin{equation}
m^l_N := \av{M^l} = \frac{1}{Z_N} \sum_{\{ ({\bf q},{\bf b}) \}} M^l e^{-\beta H}, \qquad l=1,2,3,4.
\end{equation}
Upon introducing the free-energy density as $
{\cal F}_{N}:= (1/N)\log Z_N$,
the order parameters for an intrinsically biaxial system composed by $N$ molecules can be calculated by direct differentiation as follows
\begin{equation}
\label{eq:orderpar}
m^1_N= \der{{\cal F}_{N}}{x}\quad , \quad  m^2_N = \der{{\cal F}_{N}}{y} \quad , \quad  m^3_N = \der{{\cal F}_{N}}{z} \quad , \quad  m^4_N = \der{{\cal F}_{N}}{w} \, ,
\end{equation}
where $m^l_N=m^l_N (x,y,z,w,t; \lambda)$, for $l=1,2,3,4$.
Equation~(\ref{eq:heat}) implies that the free-energy density satisfies the following differential identity
\begin{align}
\notag
\der{{\cal F}_{N}}{t} &=    \left(\der{{\cal F}_{N}}{x}\right)^2+\der{{\cal F}_{N}}{x}\der{{\cal F}_{N}}{y}+\left(\der{{\cal F}_{N}}{y}\right)^2+ \lambda \left[
 \left(\der{{\cal F}_{N}}{z}\right)^2+\der{{\cal F}_{N}}{z}\der{{\cal F}_{N}}{w}+\left(\der{{\cal F}_{N}}{w}\right)^2\right]  \\
\label{eq:pdeFN}
&+\frac{1}{N} \left[ \dersec{{\cal F}_{N}}{x} + \dermixd{{\cal F}_{N}}{x}{y} + \dersec{{\cal F}_{N}}{y} + \lambda \left(\dersec{{\cal F}_{N}}{z} + \dermixd{{\cal F}_{N}}{z}{w} + \dersec{{\cal F}_{N}}{w} \right)\right] .
\end{align}
In Section~\ref{sec:thermodynamic limit and eos}, we derive the equations of state in the thermodynamic (large $N$) regime via a direct asymptotic approximation of the solution to equation~(\ref{eq:pdeFN}). Before proceeding, it is worth to emphasise that the case $\lambda=0$ implies a reduction of the model (\ref{eq:pdeFN}) to the one studied in~\cite{dgm nematics}, although the initial condition considered in that work depends on the intrinsic molecular biaxiality parameter $\Delta$, differently from the present case in which the degree of biaxiality of the interaction is entirely contained in the internal energy term. The  differences in the two treatments arise as  in this paper we are working with two order tensors, while in~\cite{dgm nematics} the so-called geometric approximation on the interaction potential allowed to work with a single order tensor, that is a linear combination of ${\bf q}$ and ${\bf b}$.


\section{Thermodynamic limit and equations of state}
\label{sec:thermodynamic limit and eos}
The thermodynamic limit is defined as the regime where the number of particles $N$ is large, i.e. $N \to \infty$. Under the assumption  that the free-energy admits the expansion of the form ${\cal F}_{N} = F + O \left (1/N \right)$ and by using  Eq.~(\ref{eq:pdeFN}) we obtain, at the leading order, the following Hamilton-Jacobi type equation
\begin{equation}
\label{eq:HJF}
\der{F}{t} =   \left( \der{F}{x} \right)^{2} +  \der{F}{x} \der{F}{y} + \left(\der{F}{y} \right)^{2} + \lambda \left[ \left( \der{F}{z} \right)^{2} +  \der{F}{z} \der{F}{w} + \left(\der{F}{w} \right)^{2} \right].
\end{equation}
A similar asymptotic expansion for the order parameters $m^l_N = m^l + O(1/N)$ implies the relations 
\[
m^1 = \der{F}{x} \quad, \quad m^2 = \der{F}{y} \quad,  \quad  m^3 = \der{F}{z}\quad, \quad m^4 = \der{F}{w}\,.
\]
Equation~\eqref{eq:HJF} is completely integrable and can be solved via the method of  characteristics. In particular, the solution can be expressed via the free-energy functional
\begin{align}
\notag
F &= x m^1 \!+ y m^2\!+z m^3 \!+w m^4 \! \\
\label{eq:Fsol}
&+t \left[(m^1)^2 \! + m^1 m^2\! +(m^2)^2\! +\lambda \left( (m^3)^2 \!+ m^3 m^4 \!+(m^4)^2\right)\right] \\ 
\notag
&+ S(m^1,m^2,m^3,m^4)\,,
\end{align}
where $m^1$, $m^2$, $m^3$ and $m^4$ are stationary points of the free-energy, i.e. 
\[
\der{F}{m^l} = 0 \quad  \text {for } l=1,2,3,4\,.
\]
Equivalently,  order parameters are solutions to the following system of equations
\begin{gather}
\begin{aligned}
\Psi_1:=x + (2 m^1 + m^2) t +  \der{S}{m^1}&=0, \qquad  \Psi_2:= y + (m^1 +2 m^2) t+ \der{S}{m^2}=0\,,\\
\label{eq:ESgen}
\Psi_3:=z + (2 m^3 + m^4)\lambda t +\der{S}{m^3} &= 0, \quad \, \,\, \, \,\Psi_4:= w +  (m^3 +2 m^4)  \lambda t +\der{S}{m^4}=0\,.
\end{aligned}
\end{gather}
The term $S(m^1,m^2,m^3,m^4)$ represents the entropy of the system and, as discussed below, is uniquely fixed via the initial condition $F_0 =F(x,y,z,w,t=0)$.

The system (\ref{eq:ESgen}) represents the set of equations of state for the $\lambda$-model. Hence, phase transitions can be studied through the analysis of critical points of the equations (\ref{eq:ESgen}).
Similarly to the thermodynamic models studied in~\cite{DeNittis,AntonioAnnals,BarraMoro,gvsym}, order parameters $m^l$ can be viewed as solutions to a nonlinear integrable system of hydrodynamic type where coupling constants $x$, $y$, $z$, $w$ and $t$ play the role of, respectively, space and time variables. In this framework, state curves within the critical region of a phase transition are the analog of shock waves of the hydrodynamic flow.
In order to  specify completely  equations of state (\ref{eq:ESgen}) we have to determine the function $S(m^1,m^2,m^3,m^4)$. We proceed by evaluating Eqs.~(\ref{eq:ESgen}) at $t=0$, that is
\begin{align}
\notag
x(m^{1}_{0},m^{2}_{0},m^{3}_{0},m^{4}_{0})  &= -\left . \der{S}{m^1} \right |_{m^l =m^l_0}, \quad y(m^{1}_{0},m^{2}_{0},m^{3}_{0},m^{4}_{0}) = -\left . \der{S}{m^2} \right |_{m^l =m^l_0},\\
\label{eq:ESgen0}
z(m^{1}_{0},m^{2}_{0},m^{3}_{0},m^{4}_{0})  &= -\left . \der{S}{m^3} \right |_{m^l =m^l_0}, \quad w(m^{1}_{0},m^{2}_{0},m^{3}_{0},m^{4}_{0}) = -\left . \der{S}{m^4} \right |_{m^l =m^l_0},
\end{align}
where $m^l_0 = m^l(x,y,z,w,t=0)$, with $l=1,2,3,4$. Equations (\ref{eq:ESgen0}) show that the function $S(m^1,m^2,m^3,m^4)$ can be obtained, locally, by expressing $x$, $y$, $z$ and $w$ as functions of the order parameters $m^l$ evaluated at $t=0$ and then integrating Eqs.~(\ref{eq:ESgen0}). Indeed, observing that the initial condition for $F$ is $F_0 = {\cal F}_{N,0} = (1/N) \log Z_{0,N}$, where $Z_{0,N}$ is given in (\ref{eq:partition0}), the required functions can be obtained by inverting the system
\begin{equation}
\label{minvert}
m^1_0 = \der{F_0}{x}(x,y,z,w) \, , \, m^2_0 = \der{F_0}{y}(x,y,z,w)\, , \, m^3_0 = \der{F_0}{z}(x,y,z,w)\, , \, m^4_0 = \der{F_0}{w}(x,y,z,w)\,.
\end{equation}
More explicitly, equations (\ref{minvert}) read as follows
\begin{equation}
\label{roots}
\sum_{i=1}^{6} \left(m^{l}_0 - \Lambda^{l,i} \right) X^{\Lambda^{1,i}} Y^{\Lambda^{2,i}}Z^{\Lambda^{3,i}} W^{\Lambda^{4,i}} = 0 \, \, \, ,\,l=1,2,3,4\,,
\end{equation}
where we have introduced the notation $X = \exp (x)$, $Y = \exp( y)$, $Z = \exp( z)$, $W = \exp( w)$. Hence, equations of state~(\ref{eq:ESgen}) for the  model with external fields are completely determined in terms of the roots of  system of equations~(\ref{roots}). We should also emphasise that  system (\ref{roots}) is algebraic with respect to the variables $X$, $Y$, $Z$ and $W$.

{\bf Remark.} The order parameters introduced here are related to the scalar order parameters adopted in \cite{svd2003} by the following linear transformation
\begin{equation}
	\label{eq:orderParametersMap}
	m^1=T-S/3\,,\quad m^2=-T-S/3\,,\quad m^3=T'-S'/3\,,\quad m^4=-T'-S'/3\,,
\end{equation}
where $S,T,S',T'$ are the scalar order parameters characterising the tensors ${\bf Q}:=\av{{\bf q}}$ and ${\bf B}:=\av{{\bf b}}$ in their common eigenframe, once the thermodynamic limit is performed. Specifically, by considering  the eigenframe $(\vec{e}_x,\vec{e}_y,\vec{e}_z)$, the order tensors can be written as
\begin{align}
	{\bf Q}&= S \left(\vec{e}_z \otimes \vec{e}_z - \frac{1}{3} { \bf I} \right) + T \left(\vec{e}_x \otimes \vec{e}_x - \vec{e}_{y} \otimes \vec{e}_{y}\right)\\
	{\bf B}&= S' \left(\vec{e}_z \otimes \vec{e}_z - \frac{1}{3} { \bf I} \right) + T' \left(\vec{e}_x \otimes \vec{e}_x - \vec{e}_{y} \otimes \vec{e}_{y}\right)\,.
\end{align}
The inverse of the linear transformation \eqref{eq:orderParametersMap} is 
\begin{equation}
	\label{eq:orderParametersMapInverse}
	S=-\frac{3}{2}(m^1 + m^2)\,,\quad T=\frac{1}{2}(m^1 - m^2)\,,\quad S'=-\frac{3}{2}(m^3 + m^4)\,,\quad T'=\frac{1}{2}(m^3 - m^4)\,.
\end{equation}
In \cite{gdmromanov2005} it is claimed that, in the absence of external fields, reductions $T=S'=0$, or  $T=\pm S$ and $S'=\pm 3 T'$ hold, these latter obtained  by swapping the axes of the reference frame, $\vec{e}_x,\vec{e}_y,\vec{e}_z$. These conditions read as $m^1=m^2=-S/3$ and $m^3=-m^4=T'$.

In the next section, we will introduce a new parametrisation  based on the introduction of the molecular Gibbs weights, which leads to the explicit solutions of the model.

\subsection{Equations of state}
A convenient approach to the evaluation of the entropy of the discrete model and the corresponding  equations of state starts from the statistical analysis of the `initial condition', namely the evaluation of the partition function~\eqref{eq:partition0} as a function of the external fields at $t=0$. Indeed, at $t=0$, liquid crystal molecules are mutually independent and  expectation values can be evaluated  by looking at the one-molecule partition function, 
\begin{equation}
	\label{eq:Z01}
	Z_{0,1} =  \sum_{i=1}^{6} e^{x\Lambda^{1,i} + y \Lambda^{2,i}+z\Lambda^{3,i} + w \Lambda^{4,i}} \,.
	\end{equation}
The molecular Gibbs weights \cite{Callen} at $t=0$ and as functions of the external fields take the following form 
\[ p_{0,i} (x,y,z,w):=  \frac{e^{x\Lambda^{1,i} + y \Lambda^{2,i}+z\Lambda^{3,i}+ w \Lambda^{4,i}} }{Z_{0,1}}  \quad, \quad  i=1, \dots, 6\,.\]
Notice that the partition function (\ref{eq:Z01}), ensures that the Gibbs weights fulfil the standard normalisation condition,
\begin{equation}
	\label{eq:normalisation Gibbs}
	\sum_{i=1}^{6} p_{0,i}=1 \,.
\end{equation} 
The configurational entropy of the model is standardly given by $S=-\sum_{k=1}^{6} p_k \log p_k$. At $t=0$, this reads, $S_0=-\sum_{k=1}^{6} p_{0,k} \log p_{0,k}$.    By inspection, the following  holds at $t=0$

\begin{equation}
	\label{eq:gibbs pure}
	x= \frac{1}{2} \log \frac{p_{0,1} \, p_{0,2}}{p_{0,5} \,p_{0,6}}  \quad,\quad y= \frac{1}{2} \log \frac{p_{0,3} \, p_{0,4}}{p_{0,5} \, p_{0,6}}  \quad,\quad z=\frac{1}{2} \log \frac{p_{0,3}}{p_{0,4}}  \quad,\quad w=\frac{1}{2} \log \frac{p_{0,2}}{p_{0,1}} \,.
\end{equation}
In the specific case of the model under study, one can verify  that only four out of six Gibbs weights are functionally  independent. Indeed, additionally to the normalisation constraint (\ref{eq:normalisation Gibbs}), one can readily verify the following
\begin{equation}
	\label{eq:Gibbs constraint}
	\prod_{k=1}^{3} p_{0,2k-1} = \prod_{k=1}^{3}p_{0,2k} \,.
\end{equation}
By using Eqs.~(\ref{eq:normalisation Gibbs}) and (\ref{eq:Gibbs constraint}) one can express $p_{0,6}$ and $p_{0,5}$ in terms of $p_{0,1}$, $p_{0,2}$, $p_{0,3}$ and $p_{0,4}$ as follows 
\begin{equation}
	\label{eq:w3 w5}
	p_{0,5}=  \frac{p_{0,2} p_{0,4} (1-\sum_{i=1}^4p_{0,i})}{p_{0,1} \,p_{0,3} +p_{0,2} \, p_{0,4}}	 \quad , \quad  p_{0,6}=\frac{p_{0,1} p_{0,3} (1-\sum_{i=1}^4p_{0,i})}{p_{0,1} \,p_{0,3} +p_{0,2} \, p_{0,4}} \,.
\end{equation}
Note that the entropy density, as well as the Gibbs weights, depend on the temperature and the fields via the scalar order parameters only (see Eq.~(\ref{eq:Fsol})). Therefore, the identities in Eqs.~(\ref{eq:w3 w5}) hold at every $t$.
The Gibbs weights are related to the order parameters $m^l$ via the  transformation $\varphi: (p_1,p_2,p_3,p_4)\in[0,1]^4 \to (m^1,m^2,m^3,m^4) \in \mathcal{D} \subset  \mathbb{R}^4$
where

\begin{gather}
	\begin{aligned}
		m^1=& \,p_1+p_2-\frac{1}{3} \,\quad, \quad m^3=\frac{(p_1 \, p_3-p_2 \, p_4)(1-p_1-p_2)+2p_3\, p_4 (p_2-p_1)}{p_1 \,p_3+p_2 \, p_4}, \\
		\label{eq:transformation p m}
		m^2=& \, p_3+p_4-\frac{1}{3} \,\quad, \quad m^4= \frac{(p_2 \, p_4- p_1 \, p_3 )(1-p_3-p_4)+ 2p_1 \,p_2 (p_3-p_4) }{p_1 \,p_3+p_2 \, p_4} \,.
	\end{aligned}
\end{gather}
The domain  $\mathcal{D}$ is identified by the following constraints
\begin{align*}-2/3 \leq m^1+m^2 \leq 1/3 \, \,\, &, \,\, \, -( 2/3 +m^1+m^2) \leq m^1-m^2 \leq 2/3  +m^1+m^2 \\
	-2 \leq  m^3-m^4 \leq 2\, \,\, &, \,\, \, -( 2/3 +m^1+m^2) \leq m^3+m^4 \leq 2/3 +m^1+m^2\,.
\end{align*}
Using the relations (\ref{eq:gibbs pure}), (\ref{eq:w3 w5}) and (\ref{eq:transformation p m}), and the observation (\ref{eq:ESgen0}) one obtains the following set of equations for $p_1$, $p_2$, $p_3$ and $p_4$ in terms of the fields and the temperature

\begin{subequations}
\begin{align}
	\label{eq:eosp 1}
	x +(2 p_1 +2 p_2 + p_3 +p_4 -1) t - \frac{1}{2} \log \left( \frac{(p_1 p_3 +p_2 p_4)^2}{p_3 p_4 (1-p_1-p_2-p_3-p_4)^2}\right)&=0\\
	\label{eq:eosp 2}
	y +(2p_3 +2 p_4 + p_1 +p_2 -1) t - \frac{1}{2} \log \left( \frac{(p_1 p_3 +p_2 p_4)^2}{p_1 p_2 (1-p_1-p_2-p_3-p_4)^2}\right)&=0\\
	\label{eq:eosp 3}
	z +\left( \frac{ p_1 p_3 (1-2p_1 +p_3-3p_4 -p_2 p_4 (1-2 p_2-3\, p_3 +p_4))}{p_1 p_3 +p_2 p_4} \right) \lambda t - \frac{1}{2} \log \left( \frac{p_3}{p_4} \right)&=0\\
	\label{eq:eosp 4}
	w +\left( \frac{p_2 p_4 (1-3p_1 +p_2 -2 p_4) -p_1 p_3 (1+ p_1 -3 p_2- 2 p_3)}{p_1 p_3 +p_2 p_4} \right) \lambda t - \frac{1}{2} \log \left( \frac{p_2}{p_1} \right) &=0\,.
\end{align}
\end{subequations}
Equations~(\ref{eq:eosp 1})-(\ref{eq:eosp 4}) can be viewed as the equations of state of the discrete $\lambda$-model subject to external fields, parametrised by $p_i$ and  intensive thermodynamic variables $x$, $y$, $z$, $w$ and $t$, which are the  control parameters of the model. 
Notice that Eqs.~\eqref{eq:eosp 1}-\eqref{eq:eosp 4} are the critical points of the  free-energy which can now be given the form 
	\begin{equation}
		\label{eq: free energy sol explicit}
		F= \! x  \, m^1 \! +y \, m^2  \!+ z \, m^3  \! + w \, m^4+  \! \frac {t}{2} \left[ \Tr{{\bf{Q}}^2 }  + \lambda \Tr{ {\bf{B}}^2 } \right]\! - \! \sum_{k=1}^6 p_k \log p_k  \! \,,
		\end{equation}
	where ${\bf Q}=\diag (m^1, m^2, -m^1-m^2)$ and ${\bf B}=\diag (m^3, m^4, -m^3-m^4)$, and  $m^l=m^l (p_1, p_2, p_3, p_4)$, with $l=1,2,3,4$, and $ p_{5,6}=p_{5,6}(p_1, p_2, p_3, p_4)$ are given by 
		Eqs.~\eqref{eq:transformation p m} and Eqs.~\eqref{eq:w3 w5}, respectively.

\subsubsection{Two-parameter reductions}
In this subsection, we will focus on the derivation of two-parameter reductions of the equations of state (\ref{eq:eosp 1})-(\ref{eq:eosp 4}). Such reductions arise naturally when considering the liquid crystal system constrained to suitable forms of external fields, including the case in which external fields are not present and the phase behaviour is entirely regulated by the mutual interactions among liquid crystal molecules and the temperature. The following holds in the absence of external fields.

\begin{lemma}
	\label{thm: reduction in p}
	In the absence of external fields, that is at $x\!=\!y\!=\!z\!=\!w\!=\!0$, solutions to the system (\ref{eq:eosp 1})-(\ref{eq:eosp 4}) are given by  one of the following   $2$-parameter reductions:
	\begin{enumerate}[i)]
		\item $ $
		\begin{enumerate}[a)]
			\item$p_4=p_2$ and $p_3=p_1$, with
			\begin{align}
				\label{eq:eos p ub a red 1}
				&(1-3 \,  p_1-3 \, p_2)\,t = \frac{1}{2} \log \left( \frac{p_1 \, p_2 \, ( 1-2 \,  p_1 -2 \, p_2)^2 }{ \left( p_1^2 + p_2^2 \right)^2 } \right)\\
				\label{eq:eos p ub a red 2}
				&(p_1-p_2)\left( 1 + \frac{4 \, p_1 \, p_2 -p_1 -p_2}{p_1^2+p_2^2}\right) \lambda \, t= \frac{1}{2} \log \left( \frac{p_2}{p_1} \right) \,.
			\end{align}
			\item $p_3=\left[(1-2\, p_1-2 \, p_2)  p_2^{\,2} \right]/ (p_1^{\, 2}+p_2^{\, 2})$ and $p_4=\left[ (1-2\, p_1-2\,p_2)p_1^{\,2}\right] / (p_1^{\, 2}+p_2^{\, 2})$, where $p_{1}$ and $p_{2}$ satisfy Eqs.~(\ref{eq:eos p ub a red 1})-(\ref{eq:eos p ub a red 2}).
			\item $p_1=\left[ (1-2\, p_3-2\,p_4)p_4^{\,2}\right] / (p_3^{\, 2}+p_4^{\, 2})$ and $p_2=\left[ (1-2\, p_3-2 \, p_4)p_3^{\,2} \right] /(p_3^{\, 2}+p_4^{\, 2})$, where $p_{3}$ and $p_{4}$ satisfy
			\begin{align}
				\label{eq:eos p ub c red 1}
				&(1-3 \,  p_3-3 \, p_4)\,t = \frac{1}{2} \log \left( \frac{p_3 \, p_4 \, ( 1-2 \,  p_3 -2 \, p_4)^2 }{ \left( p_3^2 + p_4^2 \right)^2 } \right)\\
				\label{eq:eos p ub c red 2}
				&(p_4-p_3)\left( 1 + \frac{4 \, p_3 \, p_4 -p_3 -p_4}{p_3^2+p_4^2}\right) \lambda \, t= \frac{1}{2} \log \left( \frac{p_3}{p_4} \right) \,.
			\end{align}
		\end{enumerate}
		\item 	$ $
		\begin{enumerate}[a)] 
			\item $p_3=p_2$ and $p_4=p_1$,  with
			\begin{align}
				\label{eq:eos p uu red 1}
				(1-3 \,p_1- 3 \, p_2)\,t &= \frac{1}{2} \log \left( \frac{  \left( 1- 2 \, p_1-2 \, p_2 \right)^2}{ 4 \,p_1 \, p_2} \right)\\
				\label{eq:eos p uu red 2}
				3 \,(p_1-p_2)\,  \lambda \, t&= \frac{1}{2} \log \left( \frac{p_1}{p_2} \right) \,.
			\end{align}
			\item $p_3=p_4=1/2- p_1 - p_2$, with Eqs.~(\ref{eq:eos p uu red 1})-(\ref{eq:eos p uu red 2}) holding for $p_1$ and $p_2$.
			
			\item $p_2=p_1=1/2- p_3 - p_4$, with
			\begin{align}
				\label{eq:eos p uu c red 1}
				(1-3 \,p_3- 3 \, p_4)\,t &= \frac{1}{2} \log \left( \frac{  \left( 1- 2 \, p_3-2 \, p_4 \right)^2}{ 4 \,p_3 \, p_4} \right)\\
				\label{eq:eos p uu c red 2}
				3 \,(p_4-p_3)\,  \lambda \, t&= \frac{1}{2} \log \left( \frac{p_4}{p_3}  \right) \,.
			\end{align}
		\end{enumerate}
	\end{enumerate}
\end{lemma}

\begin{proof}
	Let us consider  Eqs.~(\ref{eq:eosp 1})-(\ref{eq:eosp 4}) restricted to the condition $x=y=z=w=0$. Observing that Eq.~(\ref{eq:eosp 1}) has special solutions such that 
	\begin{align}
		2 p_1 +2 p_2 + p_3 +p_4 -1=&0\\
		(p_1 p_3 +p_2 p_4)^2-p_3 p_4 (1-p_1-p_2-p_3-p_4)^2=&0,
	\end{align} 
the above system  admits two solutions for $p_3$ and $p_4$ as functions of $p_1$ and $p_2$: one is $p_3=\left[ (1-2\, p_1-2 \, p_2)p_2^{\,2} \right] / (p_1^{\, 2}+p_2^{\, 2})$, $p_4=\left[ (1-2\, p_1-2\,p_2)p_1^{\,2} \right] /(p_1^{\, 2}+p_2^{\, 2})$, and the other is $p_3=p_4=1/2- p_1 - p_2$. Substituting the first into  Eq~(\ref{eq:eosp 2}) we obtain Eq.~(\ref{eq:eos p ub a red 1}), while the same constraints imply that Eqs.~(\ref{eq:eosp 3})-(\ref{eq:eosp 4}) reduce to Eq.~(\ref{eq:eos p ub a red 1}), thus proving  reduction i-b). If we consider the second solution instead, we obtain (\ref{eq:eos p uu red 1}) from Eqs.~(\ref{eq:eosp 2}) and (\ref{eq:eos p uu red 2}) from Eqs.~(\ref{eq:eosp 3})- (\ref{eq:eosp 4}), thus proving reduction ii) b). Similarly, Eq~(\ref{eq:eosp 2}) admits solutions such that
	\begin{align}
		2p_3 +2 p_4 + p_1 +p_2 -1=&0\\
		(p_1 p_3 +p_2 p_4)^2-p_1 p_2 (1-p_1-p_2-p_3-p_4)^2=&0 \,,
	\end{align}
which provide two solutions: one is $p_1=\left[ (1-2\, p_3-2\,p_4)p_4^{\,2} \right] / ( p_3^{\, 2}+p_4^{\, 2})$, $p_2=\left[ (1-2\, p_3-2 \, p_4)p_3^{\,2}\right] (p_3^{\, 2}+p_4^{\, 2})$, and the other is $p_2=p_1$ and $p_3=1/2- p_1 - p_4$. Substituting the first solution into  Eq~(\ref{eq:eosp 1}), one obtains (\ref{eq:eos p ub c red 1}), while the same constraints imply that Eqs.~(\ref{eq:eosp 3})-(\ref{eq:eosp 4}) reduce to Eq.~(\ref{eq:eos p uu c red 1}), that is the reduction i-c). If we consider the second solution instead, we obtain (\ref{eq:eos p uu red 1}) from Eq.~(\ref{eq:eosp 2}) and (\ref{eq:eos p uu red 2}) from Eq.~(\ref{eq:eosp 3})-(\ref{eq:eosp 4}), thus yielding reduction ii)c). 
	When  $2 \,p_1 +2 \,p_2 + p_3 +p_4 -1\neq 0$ and $2\,p_3 +2 \,p_4 + p_1 +p_2 -1\neq 0 $, we can eliminate $t$ from  Eqs.~(\ref{eq:eosp 1})-(\ref{eq:eosp 2}) to get the following 
	\begin{align}
		\notag
		\log \left( \frac{p_1 \, p_3 +p_2 \,p_4}{p_3 \, p_4}\right)&= \frac{p_1+p_2 -p_3-p_4}{1-p_1-p_2-2\,p_3-2\,p_4} \log{\left( \frac{ (1-p_1-p_2-p_3-p_4)^2}{p_1\,  p_3 +p_2 \, p_4}\right) }  \\
		\label{reduction a proof} &+ \frac{1-2\,p_1-2\,p_2 -p_3-p_4}{1-p_1-p_2-2\,p_3-2\,p_4} \, \log \left(  \frac{p_1 \,p_3 +p_2 \, p_4}{p_1 \,  p_2}\right)\,.
	\end{align}
	Let $k>0$ be an arbitrary constant and $\vartheta$ be the  scale transformation defined by  $\vartheta: \, p_i\to k \, p_i$ for $i=1,2,4,6$. The l.h.s. of Eq.~(\ref{reduction a proof}) is invariant under the action of $\vartheta$, and is therefore independent of $k$. For consistency, the r.h.s. must retain the same property. 
	By applying $\vartheta$ to Eq~(\ref{reduction a proof}) and requiring that the r.h.s. is does not dependent on $k$, one obtains that  solutions satisfy
	\begin{equation}
		\label{eq:red proof a}
		p_1+p_2=p_3+p_4 \,.
	\end{equation} 
	We proceed by eliminating the factor $\lambda \,t$ from Eqs.~(\ref{eq:eosp 3})-(\ref{eq:eosp 4}), obtaining 
	\begin{equation}
		\label{reduction a proof 2}
		\log \left( \frac{p_1}{p_2}\right)\\
		=\frac{(1-3\,p_1+p_2-2\,p_4) \,p_2 \,p_4 -(1+p_1-2\,p_3-3\,p_2)\, p_1\,p_3}{(1-2\, p_1 +p_3 -3\, p_4) \, p_1 \,p_3 -(1-2\,p_2-3\, p_3+p_4)\, p_2\,p_4} \log{\left( \frac{p_3}{p_4 }\right) } \,.
	\end{equation}
	The generic solution is obtained by the same scaling argument. More precisely, invariance of both sides of Eq~(\ref{reduction a proof 2}) under the action of $\vartheta$ gives $	(p_1 \, p_3 -p_2 \, p_4)(p_1 \, p_3 +p_2 \, p_4)(p_1+ p_4 - p_2- p_3)  \log{ \left(\frac{p_4}{p_3} \right) }=0$, which can be realised in the two following cases
	\begin{align}
		\label{eq:red proof a i}
		p_1+ p_4 &= p_2+ p_3\\
		\label{eq:red proof a ii}
		p_1 \, p_3 &=p_2 \, p_4 \,.
	\end{align}
	System of Eqs.~(\ref{eq:red proof a})-(\ref{eq:red proof a i}) has solutions $p_3=p_1$ and $p_4=p_2$, while system of Eqs.~(\ref{eq:red proof a})-(\ref{eq:red proof a ii}) has solutions $p_3=p_2$ and $p_4=p_1$. By imposing the first of the two sets of constraints to Eqs.~(\ref{eq:eosp 1})-(\ref{eq:eosp 4}), one obtains the system of equations~(\ref{eq:eos p ub a red 1})-(\ref{eq:eos p ub a red 2}), hence proving the reduction i)a), while the second set of constraints gives Eqs.~(\ref{eq:eos p uu red 1})-(\ref{eq:eos p uu red 2}), thus proving the reduction ii)a). 
\end{proof}

As we prove  in Theorem~\ref{thm: reduction m}, Lemma~\ref{thm: reduction in p} has a remarkable implication on the structure of the two order tensors of the theory. In order to proceed, it may be convenient to recall a criterion to  characterise the degree of biaxiality of a given order tensor ${\bf \Omega}$. This will be based on the \emph{biaxiality parameter} $\beta^2 ({\bf \Omega}):= 1- 6 \frac{\text{Tr}^2 ({\bf \Omega}^3)}{\text{Tr}^3 ({\bf \Omega}^2)}$, satisfying $0\leq \beta^2\leq 1$ \cite{Kaiser}.
\begin{definition} 
	A tensor ${\bf \Omega}$ is said to be uniaxial if $\beta^2 ({\bf \Omega})=0$ and biaxial if $0<\beta^2 ({\bf \Omega})\leq 1$. Furthermore,  in the extreme case $\beta^2 ({\bf \Omega})=1$,  ${\bf \Omega}$ is said to be maximally biaxial.
\end{definition}

The following theorem characterises the allowed forms of the two order tensors of the model. 
\begin{theorem}
	\label{thm: reduction m}
	In the absence of external fields, at all  temperatures and values of $\lambda$, the order tensors take one of the following two forms
	\begin{enumerate}[i)]
		\item ${\bf{Q}}$  uniaxial and ${\bf{B}}$ maximally biaxial;
		\item $\bf{Q}$ and $\bf{B}$  both uniaxial.

	\end{enumerate} 
\end{theorem}
\begin{proof}
	The result is readily obtained by considering  Lemma~\ref{thm: reduction in p} and the transformation $\varphi$ specified by Eqs.~(\ref{eq:transformation p m}). The subcases a), b) and c) of Lemma~\ref{thm: reduction in p} in each of the two cases i) and ii), correspond to a particular choice of the principal axis. For instance, the transformation $\varphi$  evaluated along with case i) a) implies ${\bf{Q}}={\text{diag}} \,(m^1,m^1,-2\,m^1)$ and ${\bf{B}}=\text{diag} \,(m^3,-m^3,0)$, with
	\[ m^2=m^1=-\frac{1}{3} + p_1 +p_2 \quad \text{ and }\quad  m^4=-m^3= (p_1-p_2) \left( 1 + \frac{4 \,p_1 \, p_2 -p_1-p_2}{p_1^2+p_2^2}\right)\,,\]
	while case ii) a) leads to ${\bf{Q}}={\text{diag}} \,(m^1,m^1,-2\,m^1)$ and ${\bf{B}}=\text{diag} \,(m^3,m^3,-2 m^3)$, with 
	\[ m^2=m^1 = -\frac{1}{3} + p_1 +p_2  \quad \text{ and } \quad m^4=m^3= p_2 -p_1\,.\]
	
	Similarly, case i) b) corresponds to ${\bf{Q}}={\text{diag}} \,(m^1,-2 \, m^1, m^1)$ and ${\bf{B}}=\text{diag} \,(m^3,0, -m^3)$, with
	\[ m^2=-2 m^1=2 \left( \frac{1}{3} - p_1  -  p_2\right) \quad \text{ and }\quad  m^3= (p_1-p_2) \left( 1 + \frac{4 \,p_1 \, p_2 -p_1-p_2}{p_1^2+p_2^2}  \right) \, ,\, \, m^4=0\,,\]
	and case ii) b) corresponds to ${\bf{Q}}={\text{diag}} \,(m^1,-2 \, m^1, m^1)$ and ${\bf{B}}=\text{diag} \,(m^3,-2\,m^3, m^3)$, with
	\[ m^2=-2 m^1=2 \left( \frac{1}{3} - p_1  -  p_2\right) \quad \text{ and }\quad  m^4=-2 \,m^3= -2 (p_1-p_2)\,.\]
	Finally, case i) c) corresponds to ${\bf{Q}}={\text{diag}} \,(-2 m^2, \, m^2, m^2)$ and ${\bf{B}}=\text{diag} \,(0,m^4 ,-m^4)$, with
	\[ m^1=-2 m^2=2 \left( \frac{1}{3} - p_3  -  p_4\right) \quad \text{ and }\quad   m^3=0  \, ,\,m^4= (p_4-p_3) \left( 1 + \frac{4 \,p_3 \, p_4 -p_3-p_4}{p_3^2+p_4^2}  \right)  \,\,,\]
	and case ii) c) leads to  ${\bf{Q}}={\text{diag}} \,(-2 m^2, \, m^2, m^2)$  and ${\bf{B}}=\text{diag} \,(-2 m^4, m^4,m^4)$, with 
	\[ m^2=m^1 = -\frac{1}{3} + p_1 +p_2  \quad \text{ and } \quad  m^3=-2 m^4 =1 -2\, p_1 -4 \, p_4\,.\]
	The statement is proven  by evaluating the biaxiality parameter $\beta^2$ for ${\bf{Q}}$ and ${\bf{B}}$ along all cases. Due to the invariance by exchange of principal axes, $\beta^2$ for ${\bf{Q}}$ and ${\bf{B}}$ will take same values for all subcases a), b) and c) of a given case. Without loss of generality, we can consider  cases  i) a) and ii) a) to get, respectively, 
	\begin{enumerate}[i)]
		\item $\beta^2 ({\bf{Q}})=1\!- \!6 \frac{\text{Tr}^2 ({ {\text{diag}} \,((m^1)^3,(m^1)^3,-8\,(m^1)^3)})}{\text{Tr}^3 ({{\text{diag}} \,((m^1)^2,(m^1)^2,4\,(m^1)^2)})}=0$ , $\beta^2 ({\bf{B}})=1\!-\! 6 \frac{\text{Tr}^2 ({ \text{diag} \,((m^3)^3,-(m^3)^3,0)})}{\text{Tr}^3 ({\text{diag} \,((m^3)^2,(m^3)^2,0)})}=1$, that is ${\bf{Q}}$ uniaxial and ${\bf{B}}$ maximally biaxial;
		\item $\beta^2 ({\bf{Q}})=1\! - \! 6 \frac{\text{Tr}^2 ({ {\text{diag}} \,((m^1)^3,(m^1)^3,-8\,(m^1)^3)} )}{\text{Tr}^3 ({{\text{diag}} \,((m^1)^2,(m^1)^2,4\,(m^1)^2)})}=0$ , $\beta^2 ({\bf{B}})=1\!-\! 6 \frac{\text{Tr}^2 ({ \text{diag} \,((
				m^3)^3,(m^3)^3,-8\,(m^3)^3)})}{\text{Tr}^3(\text{diag} \,((m^3)^2,(m^3)^2,4\, (m^3)^2) )}=0$, hence ${\bf{Q}}$ and ${\bf{B}}$ both uniaxial.
	\end{enumerate}
\end{proof}
A direct consequence of the reductions ii) in Theorem~\ref{thm: reduction m} and the transformation $\varphi$ is the following corollary. 
\begin{corollary}
	\label{thm:uniaxial-uniaxial m explicit}
	The equations of state for the model in the case of  ${\bf{Q}}$ and  ${\bf{B}}$ both uniaxial, cases ii) in Theorem~(\ref{thm: reduction in p}), can be written explicitly in terms of the eigenvalues $m^l$. In particular, we have that reductions ii) a), b) and c) can be written as follows
	\begin{enumerate}[ii)]
		\item \begin{enumerate}[a)]
			\item	$m^2=m^1$ and $m^4=m^3$ with 
			\begin{align}
				\label{eq: eos red m uu a eq1}
				6 \, m^1 t &= \log \left( \frac{(1+3\, m^1 + 3 \, m^3 )(1+ 3 \, m^1 - 3 \, m^3)}{(1-6 \, m^1)^2} \right)\,,\\
				\label{eq: eos red m uu a eq2}
				6 \, m^3 \, \lambda \, t &= \log \left( \frac{1+3\, m^1 + 3 \, m^3 }{1+ 3 \, m^1 - 3 \, m^3} \right) \,;
			\end{align}
			\item $m^2=-2\,m^1$ and $m^4=-2\,m^3$ with $m^1$ and $m^3$ specified by Eqs.~(\ref{eq: eos red m uu a eq1})-(\ref{eq: eos red m uu a eq2});
			\item		$m^1=-2\,m^2$ and $m^3=-2 \,m^3$ with 
			\begin{align}
				\label{eq: eos red m uu c eq1}
				6 \, m^2 t &= \log \left( \frac{(1+3\, m^2 + 3 \, m^4 )(1+ 3 \, m^2 - 3 \, m^4)}{(1-6 \, m^2)^2} \right)\,,\\
				\label{eq: eos red m uu c eq2}
				6 \, m^4 \, \lambda \, t &= \log \left( \frac{1+3\, m^2 + 3 \, m^4 }{1+ 3 \, m^2- 3 \, m^4} \right) \,.
			\end{align}
		\end{enumerate}
	\end{enumerate}

\end{corollary}

\begin{proof}
	As shown in the proof of Theorem~(\ref{thm: reduction m}), the transformation $\varphi$ is linear when restricted to the case of ${\bf{Q}}$ and ${\bf{B}}$ both uniaxial. Hence,  the transformation can be easily inverted  to get the projection  $\varphi^{-1}: (m^1,m^2,m^3,m^4) \longmapsto (p_1, p_2,p_3, p_4)$ along with each particular 2-parameter reduction. The equations in terms of the eigenvalues are then  obtained by application of the inverse transformation for the specific reduction to the corresponding set of equations in the $p-$variables. Taking the case  ii) a) as an  example, the application of  $\varphi_{a}^{-1}:=\{ \varphi \,|_{p_3=p_2,\, p_4=p_1}\}^{-1}$, explicitly given by
	\[  \varphi_{a}^{-1}:\, \, (m^1, m^3) \longmapsto (p_1, p_2)=\left( \frac{1+3 \, m^1 -3 \, m^3}{6} , \frac{1+3 \, m^1 + 3 \, m^3}{6} \right) \,,\]
	to Eqs.~(\ref{eq:eos p uu red 1})-(\ref{eq:eos p uu red 2}) gives Eqs.~(\ref{eq: eos red m uu a eq1})-(\ref{eq: eos red m uu a eq2}).  Equations (\ref{eq: eos red m uu a eq1})-(\ref{eq: eos red m uu a eq2}) and (\ref{eq: eos red m uu c eq1})-(\ref{eq: eos red m uu c eq2}) for b) and c), respectively, are obtained in a similar fashion.
\end{proof}
Unlike  uniaxial-uniaxial reductions discussed above, uniaxial-maximally biaxial reductions cannot be  written in explicit simple form in terms of the $m^l$ variables.

In this paper, we will focus our discussion on the phase behaviour  in the absence of external fields. We note, however, that the above reductions are also compatible with non-zero fields values subject to suitable constraints. Indeed, the following proposition allows to identify the constraints on the external fields so that the system admits \emph{uniaxial-maximally biaxial} and \emph{uniaxial-uniaxial} solutions for ${\bf{Q}}$ and ${\bf{B}}$. In such cases we can still consider 2-parameter reductions of the system, with the equations of state also accounting  for the action of applied fields. 

\begin{proposition}
	\label{thm: reductions fields}
	In the presence of external fields, the system~(\ref{eq:eosp 1})-(\ref{eq:eosp 4}) admits the following \emph{uniaxial-maximally biaxial} two-parameter reductions:
	\begin{enumerate}[i)]
		\item 
		\begin{enumerate}[a)]
			\item   $p_3=p_1$ and $p_4=p_2$, provided that  external fields satisfy $y=x$ and $w=-z$, that is $\epsilon_1=\epsilon_2$ and $\chi_3=\frac{\chi_1 +\chi_2}{2}$, specified by
			\begin{align}
				\label{eq:eos p ub a f red 1}
				x +  (3 \,  p_1+3 \, p_2-1)\,t &= \frac{1}{2} \log \left( \frac{  \left( p_1^2 + p_2^2 \right)^2 }{ p_1 \, p_2 \, ( 1-2 \,  p_1 -2 \, p_2)^2} \right)\\
				\label{eq:eos p ub a f red 2}
				z + (p_2-p_1)\left( 1 + \frac{4 \, p_1 \, p_2 -p_1 -p_2}{p_1^2+p_2^2}\right) \lambda \, t&= \frac{1}{2} \log \left( \frac{p_1}{p_2} \right) \, 
			\end{align}
			\item  $p_3=\frac{(1-2\, p_1-2 \, p_2)p_2^{\,2}}{p_1^{\, 2}+p_2^{\, 2}}$ and $p_4=\frac{(1-2\, p_1-2\,p_2)p_1^{\,2}}{p_1^{\, 2}+p_2^{\, 2}}$,  provided that $x=0$ and $z=2w$, that is $\epsilon_1=\epsilon_3$ and $\chi_2=\frac{\chi_1 +\chi_3}{2}$,  specified by
			\begin{align}
				\label{eq:eos p ub b f red 1}
				y +  (1-3 \,  p_1-3 \, p_2)\,t &= \frac{1}{2} \log \left( \frac{p_1 \, p_2 \, ( 1-2 \,  p_1 -2 \, p_2)^2 }{ \left( p_1^2 + p_2^2 \right)^2 } \right)\\
				\label{eq:eos p ub b f red 2}
				w + (p_1-p_2)\left( 1 + \frac{4 \, p_1 \, p_2 -p_1 -p_2}{p_1^2+p_2^2}\right) \lambda \, t&= \frac{1}{2} \log \left( \frac{p_2}{p_1} \right) \, 
			\end{align}
			\item $p_1=\frac{(1-2\, p_3-2\,p_4)p_4^{\,2}}{p_3^{\, 2}+p_4^{\, 2}}$ and $p_2=\frac{(1-2\, p_3-2 \, p_4)p_3^{\,2}}{p_3^{\, 2}+p_4^{\, 2}}$, provided that $y=0$ and $w=2z$, that is $\epsilon_2=\epsilon_3$ and $\chi_1=\frac{\chi_2 +\chi_3}{2}$, and specified by
			\begin{align}
				\label{eq:eos p ub c red 1}
				x+(1-3 \,  p_3-3 \, p_4)\,t &= \frac{1}{2} \log \left( \frac{p_3 \, p_4 \, ( 1-2 \,  p_3 -2 \, p_4)^2 }{ \left( p_3^2 + p_4^2 \right)^2 } \right)\\
				\label{eq:eos p ub c red 2}
				z+ (p_4-p_3)\left( 1 + \frac{4 \, p_3 \, p_4 -p_3 -p_4}{p_3^2+p_4^2}\right) \lambda \, t&= \frac{1}{2} \log \left( \frac{p_3}{p_4} \right) \,  
			\end{align}
		\end{enumerate}
		and  \emph{uniaxial-uniaxial} two-parameter reductions:

		\item \begin{enumerate}[a)]
			\item   $p_3=p_2$ and $p_4=p_1$,   provided that $x=y$ and $z=w$, that is $\epsilon_1=\epsilon_2$ and $\chi_1=\chi_2$,  specified by
			\begin{align}
				\label{eq:eos p uu red f 1}
				x+ (3 \,p_1 +3 \, p_2-1)\,t &= \frac{1}{2} \log \left( \frac{4 \,p_1 \, p_2 }{ \left( 1- 2 \, p_1-2 \, p_2 \right)^2 } \right)\\
				\label{eq:eos p uu red f 2}
				z+3 \,(p_2-p_1)\,  \lambda \, t&= \frac{1}{2} \log \left( \frac{p_2}{p_1} \right) \, 
			\end{align}
			\item  $p_3=p_4=\frac{1}{2}- p_1 - p_2$, provided by $x=0$ and $z=0$, that is $\epsilon_1=\epsilon_3$ and $\chi_1=\chi_3$, specified by
			\begin{align}
				\label{eq:eos p uu b red f 1}
				y+ (1-3 \,p_1- 3 \, p_2)\,t &= \frac{1}{2} \log \left( \frac{  \left( 1- 2 \, p_1-2 \, p_2 \right)^2}{ 4 \,p_1 \, p_2} \right)\\
				\label{eq:eos p uu b red f 2}
				w+3 \,(p_2-p_1)\,  \lambda \, t&= \frac{1}{2} \log \left( \frac{p_2}{p_1} \right) \, 
			\end{align}
			
			\item $p_2=p_1=\frac{1}{2}- p_3 - p_4$, provided that $y=0$ and $w=0$, that is $\epsilon_2=\epsilon_3$ and $\chi_2=\chi_3$, specified by
			\begin{align}
				\label{eq:eos p uu c red f 1}
				x+	(1-3 \,p_3- 3 \, p_4)\,t &= \frac{1}{2} \log \left( \frac{  \left( 1- 2 \, p_3-2 \, p_4 \right)^2}{ 4 \,p_3 \, p_4} \right)\\
				\label{eq:eos p uu c red f 2}
				z+	3 \,(p_3-p_4)\,  \lambda \, t&= \frac{1}{2} \log \left( \frac{p_3}{p_4}  \right) \,.
			\end{align}
		\end{enumerate}
	\end{enumerate}
	\end{proposition}
	\begin{proof}
	The constraints on the fields follow from looking for either  uniaxial-maximally biaxial or uniaxial-uniaxial reductions of the whole set of equations of state, Eqs.~(\ref{eq:eosp 1})-(\ref{eq:eosp 4}). For instance, the uniaxial-maximally biaxial reduction i) a) requires $p_3=p_1$ and $p_4=p_2$. Eqs.~ (\ref{eq:eosp 1})-(\ref{eq:eosp 4}) restricted to this constraint imply that  compatibility of the first two equations restricts fields to $y=x$, while compatibility of the third and fourth requires $w=-z$. Elementary algebraic manipulations lead to  Eqs.~(\ref{eq:eos p ub a f red 1})-(\ref{eq:eos p ub a f red 2}). The set of field-dependent equations of state in the case i) b) and c), and ii) a), b) and c) are obtained following the same procedure, together with the associated constraints on the fields.
\end{proof}


\section{Order parameters in the 2-parameter reductions} 
\label{sec:order parameters}
The equations of state \eqref{eq:eosp 1}-\eqref{eq:eosp 4} are the critical points of the free-energy functional \eqref{eq: free energy sol explicit}.  According to the definition of free-energy  adopted in this paper, global maxima identify stable states of the system and associated phases.  Coexistence curves (hypersurfaces, in  general) arise as sets of control parameters for which two or more local maxima are resonant, hence identifying the coexistence of the corresponding phases. In order to proceed, we should therefore first identify all local maxima for each choice of  control parameters $x, y, z,w, \lambda$ and $t$.
In this section, we focus on the complete characterisation of the system in the absence of external fields, i.e. $x\!=y\!=\!z\!=w\!=\,0$, hence relying on Theorem~\ref{thm: reduction m} and its implications. An immediate consequence is that the onset of phase transitions is determined by analysing the singularities of  2-dimensional maps defined by Eqs.~\eqref{eq:eos p ub a red 1}-\eqref{eq:eos p uu c red 2} (see  \cite{dgm nematics} for an exhaustive treatment). Noticeably, this is a more affordable task  compared to  the full 4-dimensional problem governing the system when external fields are present, Eqs.~ \eqref{eq:eosp 1}-\eqref{eq:eosp 4}.\\
  By evaluating the Hessian matrix of the free energy density \eqref{eq: free energy sol explicit}, it turns out that none of the critical points in the uniaxial-uniaxial reductions, cases ii)a)-c) in Lemma~\ref{thm: reduction in p}) are stable. This result is consistent with what is known from  mean-field theories based on a continuum of molecular orientational states \cite{gdmromanov2005}. 
Therefore, the uniaxial-maximally biaxial  reductions, cases i)a)-c) in Lemma~\ref{thm: reduction in p}, are the only ones relevant from the equilibrium thermodynamics  viewpoint. The following subsections focus on the analysis of the resulting phase diagram  and the associated order parameters behaviour, with  case i)a) being considered for such purpose. Cases i)b) and i)c) can be straightforwardly  obtained from case i)a) via suitable linear transformations on $m^1$ and $m^3$, which merely correspond to   permutations of axes.

\subsection{Phase behaviour  in the absence of external fields}

The phase diagram of the model in the absence of external fields is shown in Fig~\ref{Fig:phase diagram lambda t}. The top row shows  the phase diagram in the $\lambda-t$ plane, while the bottom row shows the phase diagram in the  $\lambda-T^*$ plane, where $T^*$ is the dimensionless temperature defined by $T^*:=1/t=(k_B T)/\mu $. The $\lambda$-$t$ plane is divided in three regions identifying three distinct macroscopic phases, namely the \emph{isotropic} (I), the \emph{uniaxial nematic} (U) and the \emph{biaxial nematic} (B). The lines separating the different regions are either dotted black lines or  solid black lines. The former identify the so-called second-order transition lines, that is  the lines associated with phase changes characterised by continuous order parameters but discontinuous derivatives. The latter are instead associated to first-order lines, that is  the order parameters and their derivatives experience  a discontinuity when the line is crossed. Similarly to the   analysis performed in \cite{dgm nematics}, second-order lines are identified by cusp points of two-dimensional maps.  The cusp points of the model (red lines in Fig.~\ref{Fig:phase diagram lambda t}) are given explicitly in terms of the  transcendental curve
\[ 	 \mathcal{C}=\Big\lbrace{ (\lambda, t)\in \left[ 0,1 \right]\times \left[ 0,+ \infty \right)\, | \,
e^{t-\frac{1}{\lambda}}\left( 2- \lambda \, t \right) + 1- 2  \lambda\, t=0   \Big\rbrace }\,.
\] 
Notice that the cusp set can be seen as the union of two curves intersecting at the point $(\lambda,t)=(1/3,3)$.  The model admits two tricritical points, $(\lambda_{\text{tc}}^{(UB)},t_{\text{tc}}^{(UB)})=(0.217,2.854)$ (red circle) and $\left(\lambda_{\text{tc}}^{(IB)},t_{\text{tc}}^{(IB)}\right)=(2/3,3/2)$ (blue circle). The three phases coexist at the triple point, $\left(\lambda_{\text{tp}},t_{\text{tp}}\right)=(0.234,2.773)$ (green circle) identified by  the resonance condition for the corresponding maxima. A closer look in the region surrounding the triple point and the uniaxial-biaxial tricritical point is provided in the  right column.  The  cusp points in the $\lambda-T^*$ plane are given by the set 
\begin{equation*}
	 \mathcal{C^*}=\Big\lbrace{ (\lambda, T^*)\in \left[ 0,1 \right] ^2 \, | \, e^{\frac{1}{T^*}-\frac{1}{\lambda}}\left( 2- \frac{\lambda}{T^*}\right) + 1-\frac{2 \lambda}{T^*}=0  \Big\rbrace }.
	 \end{equation*}
The constraint $\lambda= T^*$  identifies the subset of cusp points associated to second-order lines for $\lambda \geq 2/3$.

\begin{figure}[h!]
	\begin{center}
		\includegraphics[height=6.8cm]{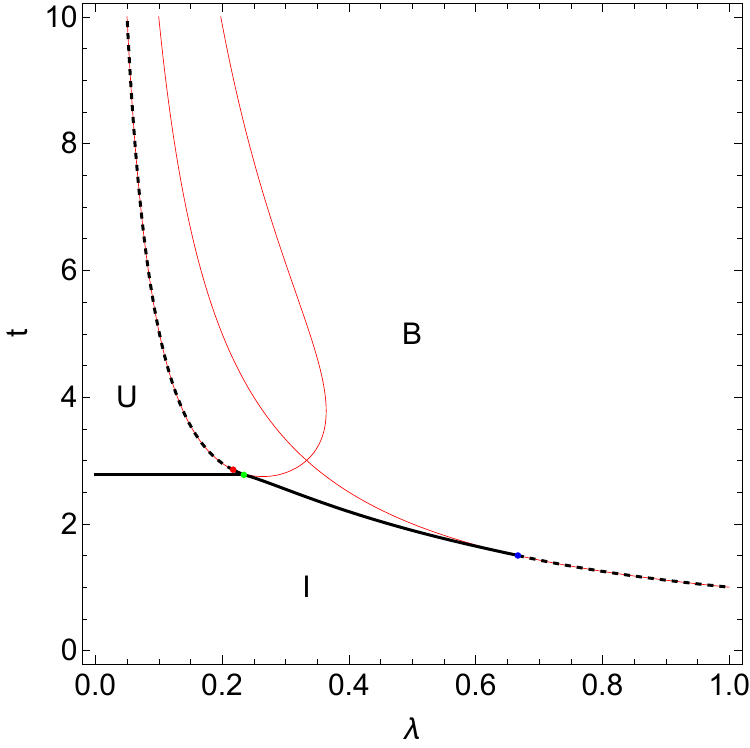} $\, \,$
		\includegraphics[height=6.8cm]{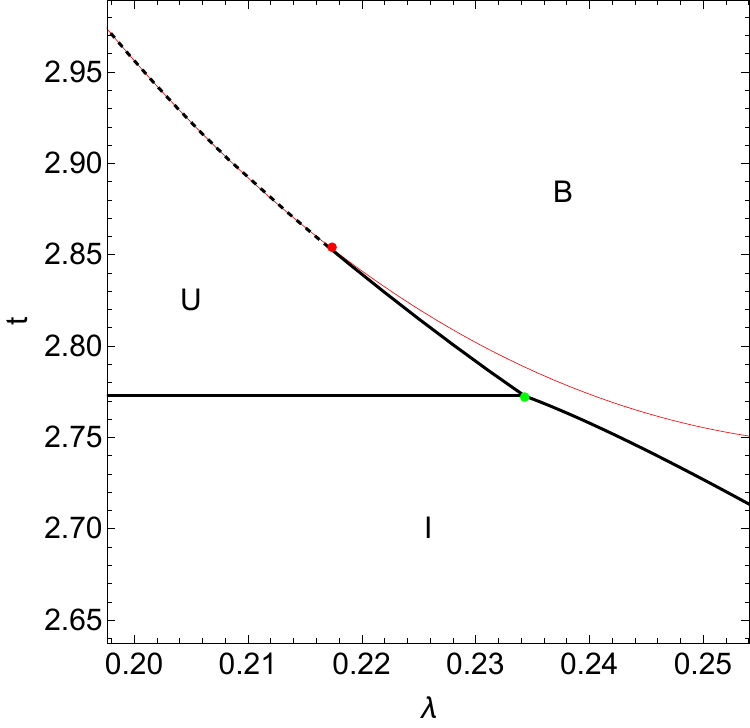}\\
				\includegraphics[height=6.8cm]{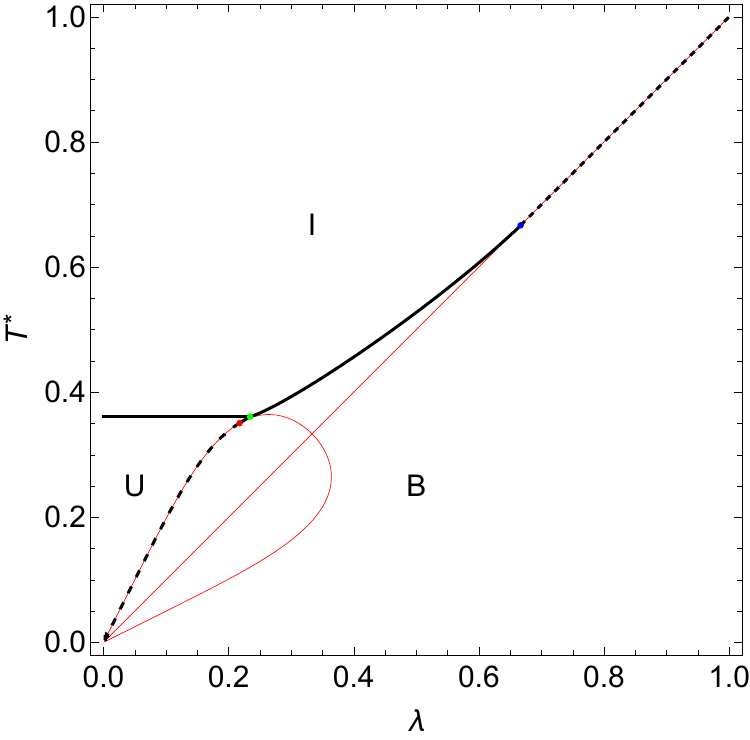} $\, \,$
			\includegraphics[height=6.8cm]{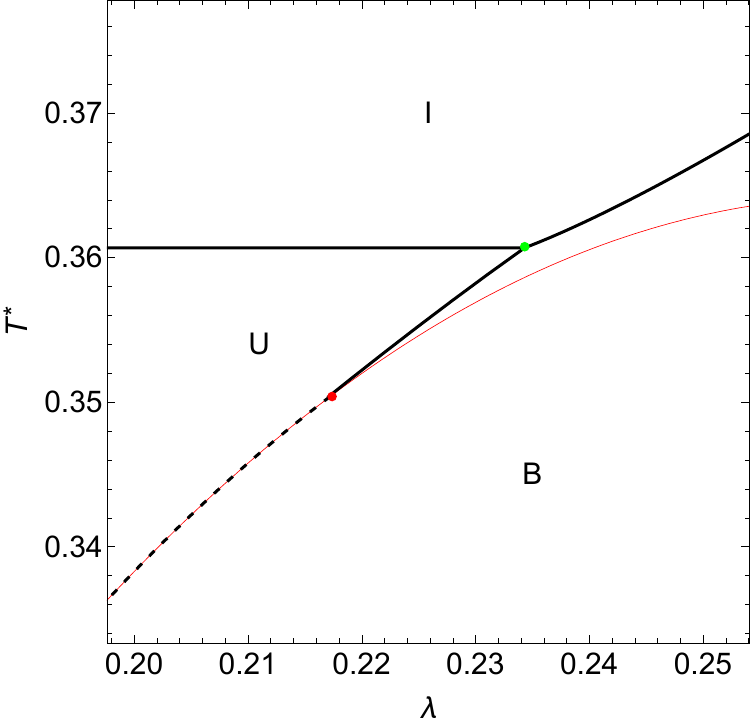}
		\caption{
			\label{Fig:phase diagram lambda t} Zero-fields phase diagram.  {\emph Top row}: phase diagram in the $\lambda-t$ plane.  {\emph Bottom row}: phase diagram in the $\lambda-T^*$ plane. The  figures in the right column represent a magnification of the phase diagram in the region surrounding the triple point (green circle) and the uniaxial-biaxial tricritical point (red circle). The line associated with uniaxial cusp points is indicated in red. }
	\end{center}
\end{figure}

\subsection{Order parameters in the absence of external fields}

In this section, we analyse the order parameters behaviour  for the reduction $m^2=m^1$, $m^4=-m^3$ in the absence of external fields as the temperature changes. Following our discussion on the phase diagram displayed in Fig~\ref{Fig:phase diagram lambda t}, we proceed by showing the expectation values $m^1$ and $m^3$ at different increasing values of $\lambda$. The values chosen for $\lambda$ aim at displaying the whole phenomenology predicted by the phase diagram. 

Fig.~\ref{Fig:mvst low lambda} shows the behaviour of order parameters in the absence of external fields for small values of $\lambda$. The case $\lambda=0$ (left column) reproduces the phenomenology of  the standard Maier-Saupe model, with the biaxial order parameter vanishing and a discontinuous isotropic-to-uniaxial nematic phase transition at $t_c^{NI}=4\log 2$. For small values of $\lambda$ (central column), that is $0<\lambda<\lambda_{tc}^{UB} \approx 0.217$, additionally to the isotropic-to-nematic phase transition, a continuous phase change at lower temperatures  is displayed from the uniaxial phase  to the biaxial phase. Consistently with the phase diagram in Fig~\ref{Fig:phase diagram lambda t}, the uniaxial-to-biaxial  phase transition becomes first-order at the uniaxial-biaxial tricritical point (right column), where both order parameters experience a gradient catastrophe at $t=t_{tc}^{(UB)}=2.854$.

\begin{figure}[h!]
	\begin{center}
		\includegraphics[height=4.6cm]{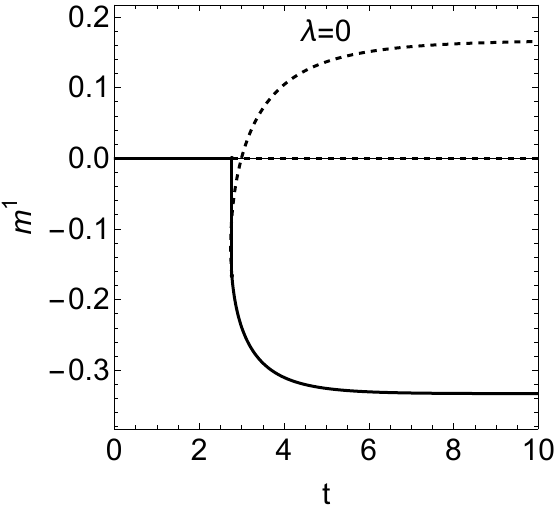}
		\includegraphics[height=4.6cm]{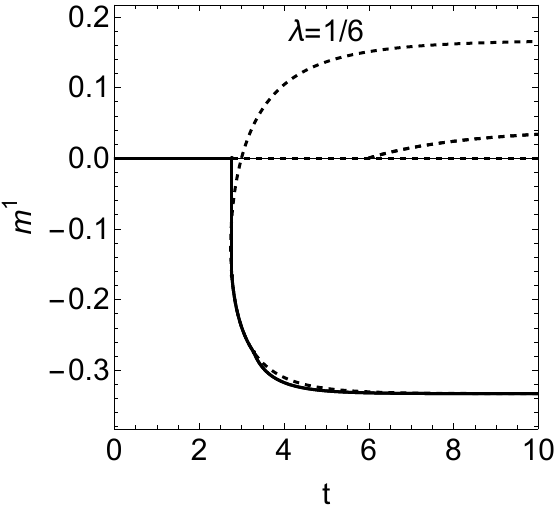}
		\includegraphics[height=4.6cm]{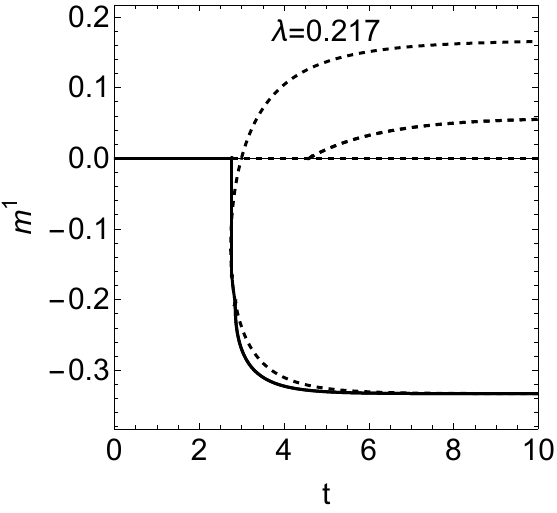}\\
		\includegraphics[height=4.6cm]{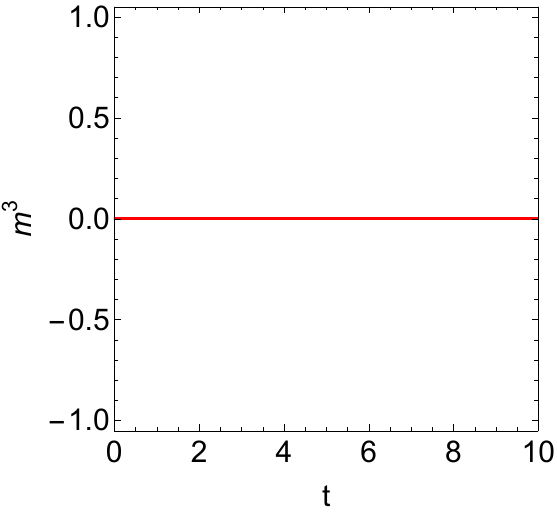}
		\includegraphics[height=4.6cm]{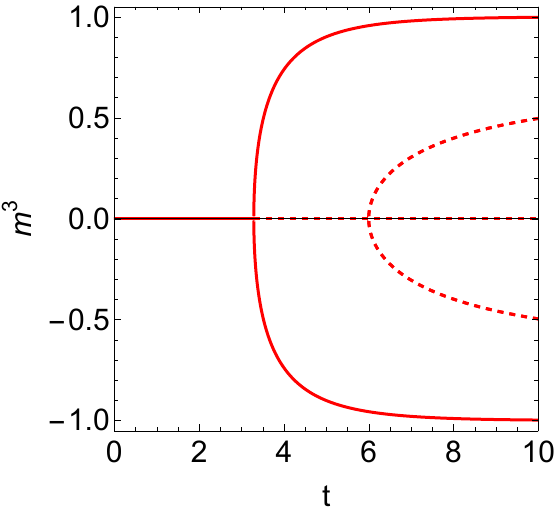}
		\includegraphics[height=4.6cm]{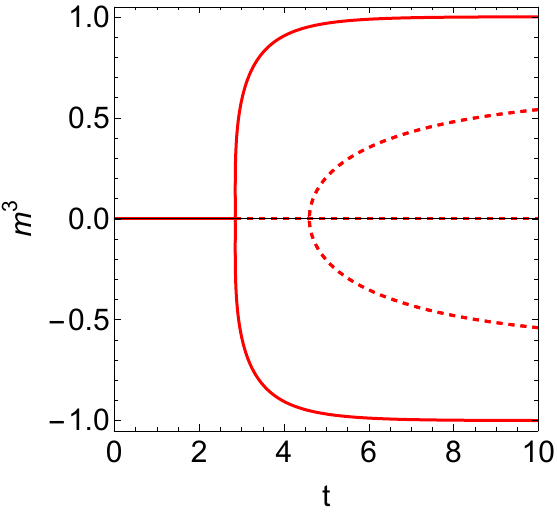}
		\caption{ Order parameters in the 2-parameter reduction for small values of $\lambda$. Each column shows both order parameters, $m^1$ (black) and $m^3$ (red) versus $t$ at a specific value of $\lambda$. The solutions corresponding to a global maximum of the free energy are displayed with solid lines, while other solutions are indicated with dotted lines. \emph{Left column}:  $\lambda=0$, that is the uniaxial Maier-Saupe  interaction potential. \emph{Centre column}:  $\lambda=1/6$. \emph{Right column}:  $\lambda=\lambda_{tc}^{UB}=0.217$.
		\label{Fig:mvst low lambda}}       
	\end{center}
\end{figure}

The behaviour for values of $\lambda$ in the interval $\left(\lambda_{tc}^{(UB)},\lambda_{tp}\right)$  is displayed in Fig~\ref{Fig:mvst double shock}. For values of $\lambda$ in this range the model predicts two first-order phase transitions, a isotropic-to-uniaxial phase transition at high temperature (low values of $t$) followed by a uniaxial-to-biaxial phase transition at lower temperatures (higher values of $t$). While the former is associated to a  shock that is static in $\lambda$, the latter is originated at the uniaxial-biaxial tricritical point and is identified by a classical shock whose location moves from low temperatures to higher temperatures as the biaxiality parameter $\lambda$ increases.

\begin{figure}[h!]
	\begin{center}
		\includegraphics[height=4.6cm]{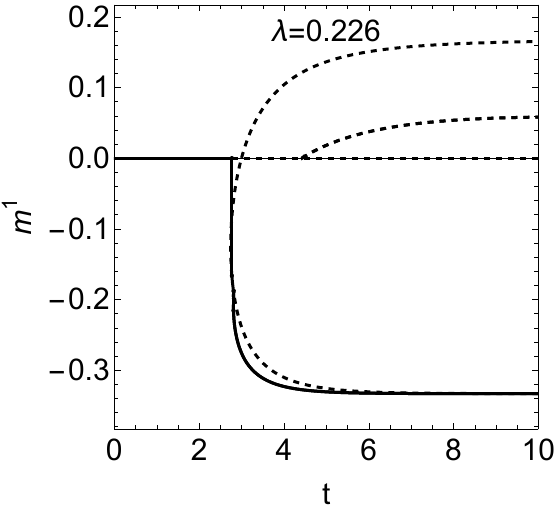}
		\includegraphics[height=4.6cm]{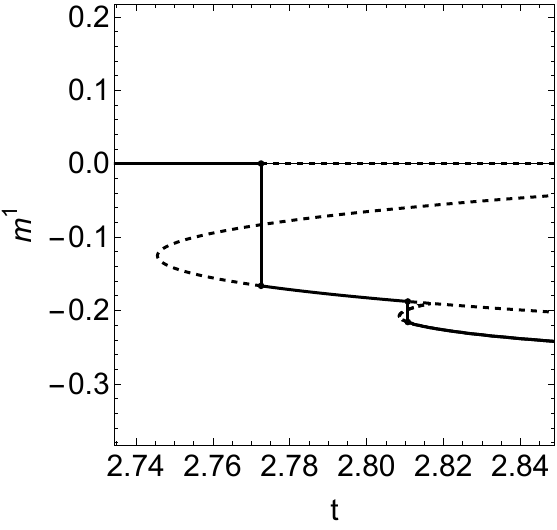} \\
		\includegraphics[height=4.6cm]{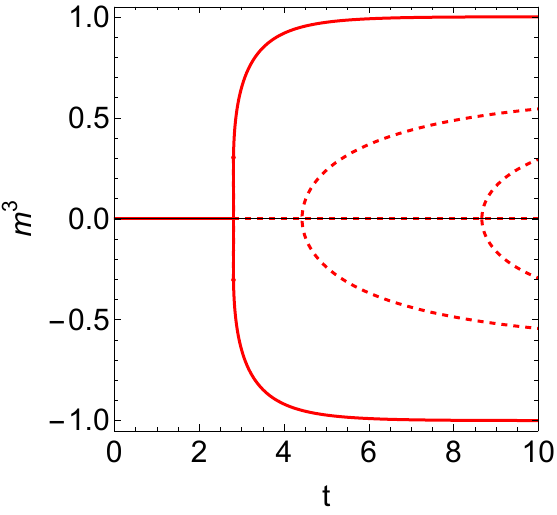}
		\includegraphics[height=4.6cm]{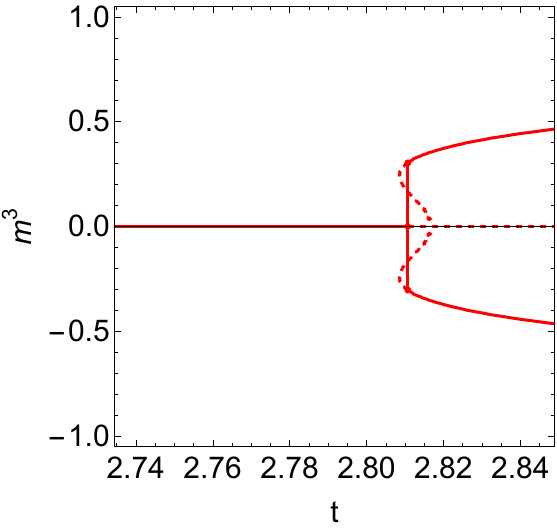}
		\caption{ Order parameters in the two-parameter reduction for values of $\lambda$ in the interval  $\lambda_{tc}^{(UB)}\leq \lambda<\lambda_{tp}$. The two column shows both order parameters versus $t$ for $\lambda=0.226 \in \left(\lambda_{tc}^{(UB)},\lambda_{tp}\right)$ as an example. The right column displays a magnification of the order parameters around the  tricritical temperature and triple point temperature.
			\label{Fig:mvst double shock} }  
	\end{center}
\end{figure}
 As shown in Figs.~\ref{Fig:phase diagram lambda t}  and ~\ref{Fig:mvst intermediate lambda} (left column),  for $\lambda=\lambda_{tp}=0.234$ the system displays   a triple point at which all three phases coexist. This situation is realised as the two shocks associated with the isotropic-to-uniaxial and uniaxial-to-biaxial phase transitions merge at zero external fields, giving rise to a single shock having an  amplitude given by the sum of the amplitudes of the two individual shocks. Consistently with the phase diagram in Fig.~\ref{Fig:phase diagram lambda t}, the uniaxial phase is not energetically accessible for $\lambda >\lambda_{tp}$. For instance, for $\lambda=0.284$ (centre column of Figure~\ref{Fig:mvst intermediate lambda}), the order parameters jump from the isotropic phase  to the biaxial phase as the temperature is lowered. This is also the case when $\lambda=1/3$ (right column). The case $\lambda=1/3$, as also discussed in \cite{gdmromanov2005},  leads to proportionality between the stable branches of the order parameters. Precisely this is $3 m^1 \pm m^3 =0$, corresponding to $T^{'} = \pm S$ in the convention adopted by  Virga and co-authors  in~\cite{gdmromanov2005}.  

For  values of $\lambda $ exceeding  the value $1/3$, the isotropic-to-biaxial phase transition remains first-order until a second tricritical point is disclosed. In Fig.~\ref{Fig:mvst large lambda}, the change in order of the phase isotropic-to-biaxial phase transition is displayed. For  $\lambda<\lambda_{tc}^{(IB)}$ (left column) both order parameters undergo a discontinuous jump from the isotropic solution to the biaxial one. The shock disappears when  $\lambda=\lambda_{tc}^{(IB)}=2/3$ is considered, and consequently both order parameters experience a gradient catastrophe at $t=t_{tc}^{(IB)}=3/2$. Larger values of $\lambda$ lead to a direct second-order transition from the isotropic to the biaxial phase, with the transition value given by $t^{(IB)}= 1/\lambda$. 

Our results are, both  qualitatively and quantitatively, consistent with previous studies \cite{romano2004,gdmromanov2005,preeti2011}. In particular, according to the Monte  Carlo simulation results reported in \cite{romano2004,gdmromanov2005,preeti2011}, the values of $\lambda$ at the first and second tricritical points are $ \simeq 0.24$ and $\simeq 2/3$, respectively. Moreover, the global Monte Carlo study performed in \cite{preeti2011} predicts  $\lambda \simeq 0.26$ for the triple point.
Consistency is also shown  with the extended analysis preformed in  \cite{gdmcmt2008landau}, where  a fourth degree Landau potential in the two tensors   ${\bf {Q}}$ and ${\bf {B}}$ is considered. Indeed, the Authors find a zero-fields phase diagram displaying a first- and second order isotropic-to-uniaxial-to-biaxial phase transitions and first- and second-order isotropic-to-biaxial phase transitions, thus disclosing two tricritical points and a single triple point. On the other hand, similar but not totally equivalent phase diagram topologies are described in \cite{longalandau,Mukherjee2009}, where a sixth degree Landau potential in a single order tensor has been analysed.  The phase diagram we obtain (Figure~\ref{Fig:phase diagram lambda t}) is also qualitatively in agreement with the one obtained  
	in \cite{sluckin2012}, where the Authors  derive a Landau expansion of a free energy which is intrinsically linked to a molecular-field theory, and then discuss  the  Sonnet-Virga-Durand limit. 
	
It goes without saying that  a single-tensor theory does not distinguish between intrinsic and phase biaxiality. This separation is made  clear and sharp by starting from  two molecular tensors ${\bf {q}}$ and ${\bf {b}}$ in the Hamiltonian (e.g.  \eqref{eq:H0}) which, correspondingly, give rise macroscopically to  two  order tensors   ${\bf {Q}}$ and ${\bf {B}}$  eventually accounting for the phase and intrinsic biaxiality, respectively. In this work, the  $\lambda-$model in the absence of external fields is shown to admit  two-parameter reductions, which  account for isotropic, uniaxial and intrinsic biaxial phases. Other models, as for example the Maier-Saupe model (retrieved from the $\lambda-$model setting $\lambda=0$), have been shown to produce uniaxiality and phase biaxiality  at the macroscopic level and  in the presence of external fields \cite{dgm nematics,Frisken,dunmur88,Mukherjee2013}.

\begin{figure}[h!]
	\begin{center}
		\includegraphics[height=4.6cm]{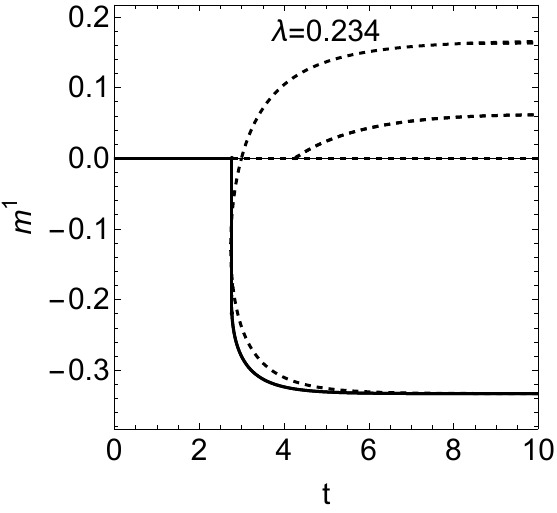}
		\includegraphics[height=4.6cm]{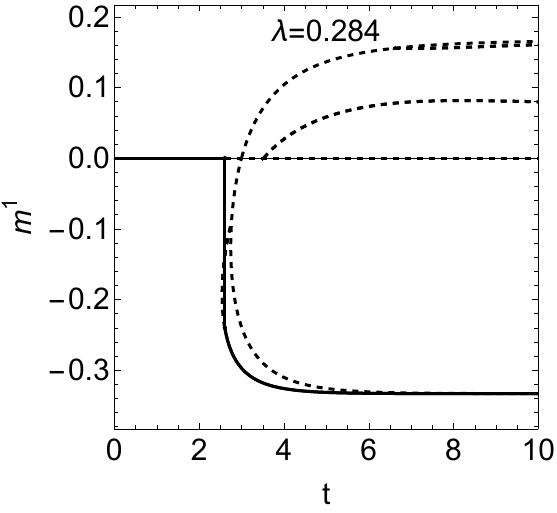}
		\includegraphics[height=4.6cm]{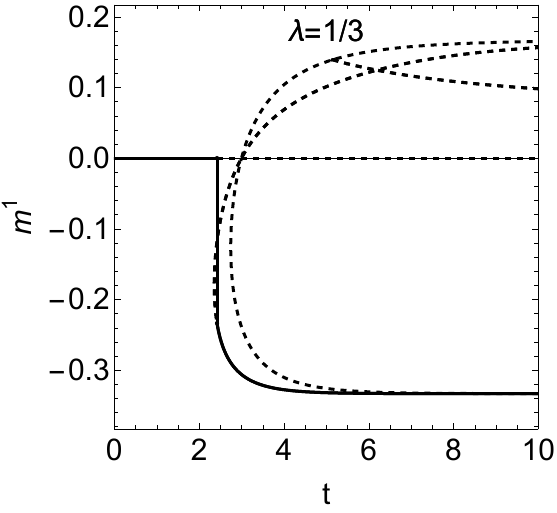}\\
		\includegraphics[height=4.6cm]{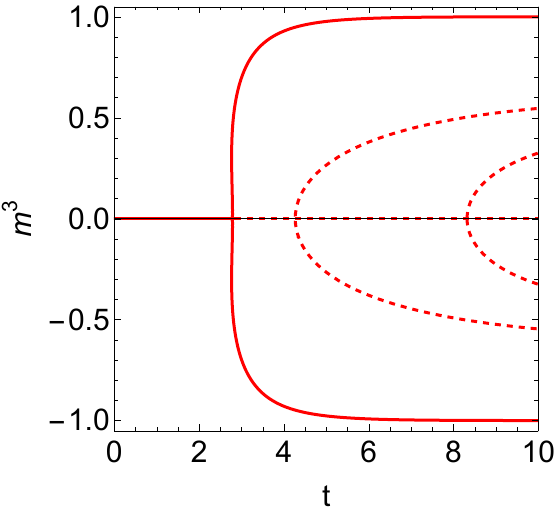}
		\includegraphics[height=4.6cm]{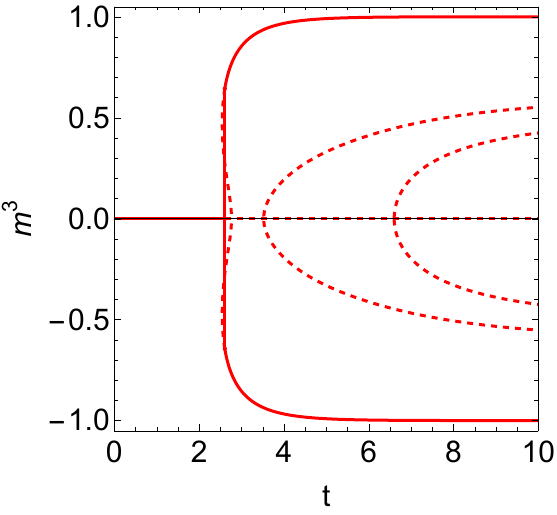}
		\includegraphics[height=4.6cm]{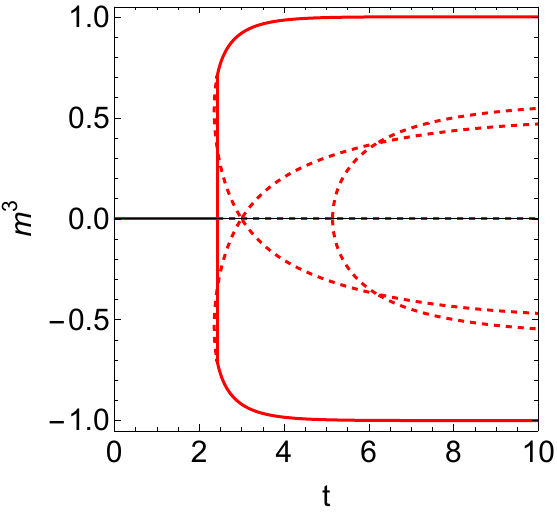}
		\caption{ Order parameters in the 2-parameter reduction for  values of $\lambda$ in the interval  $\lambda_{tp}\leq \lambda \leq 1/3$. \emph{Left column}: $\lambda=\lambda_{tp}=0.234$. \emph{Centre column}: $\lambda=0.284$. \emph{Right column}: $\lambda=1/3$. 
			\label{Fig:mvst intermediate lambda}
		}
	\end{center}
\end{figure}

\begin{figure}[h!]
	\begin{center}
		\includegraphics[height=4.6cm]{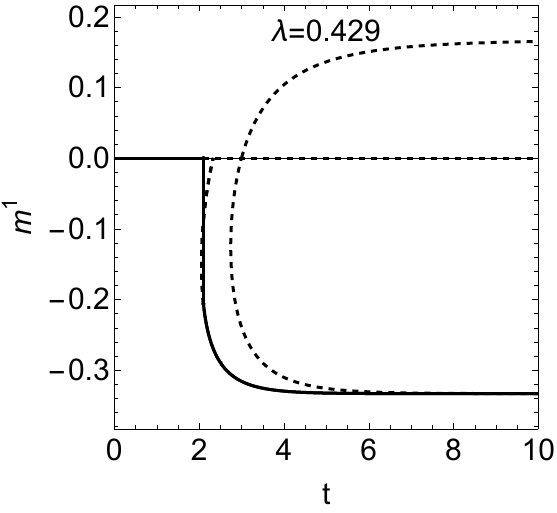}
		\includegraphics[height=4.6cm]{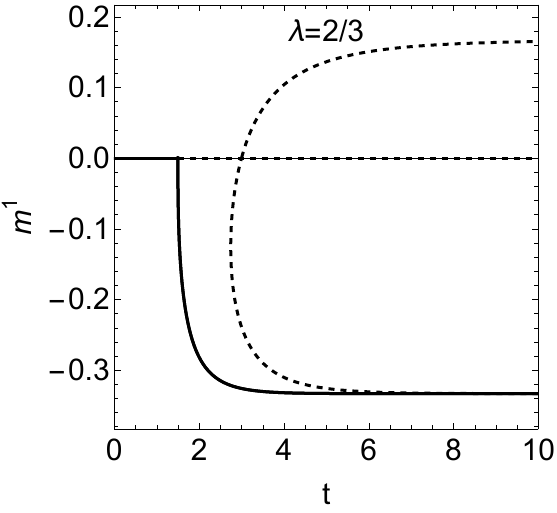}
		\includegraphics[height=4.6cm]{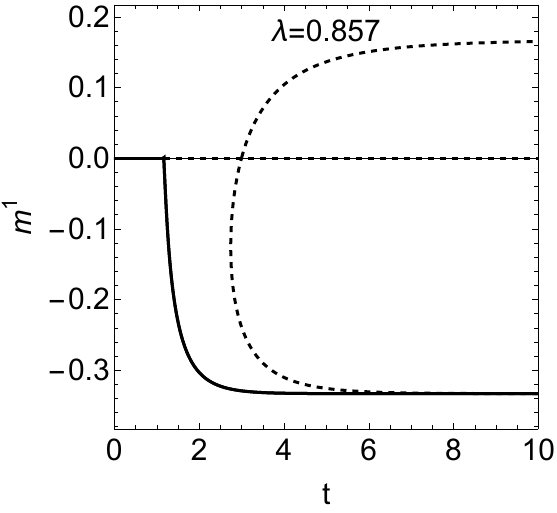}\\
		\includegraphics[height=4.6cm]{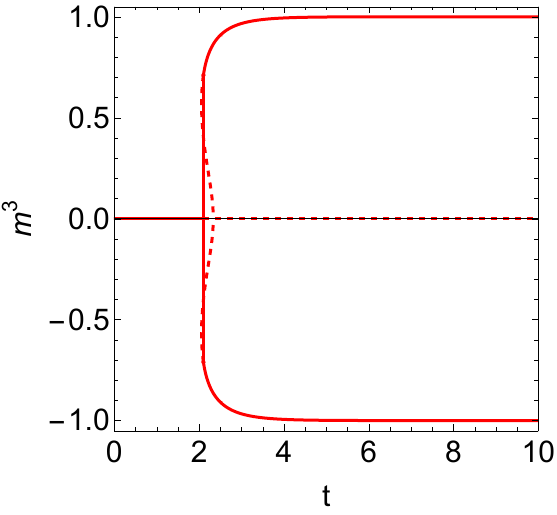}
		\includegraphics[height=4.6cm]{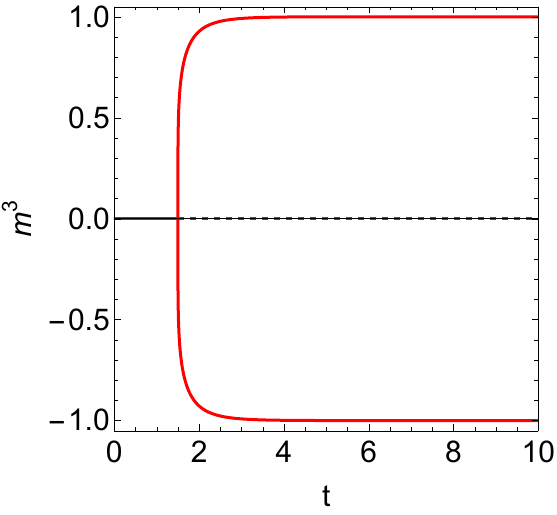}
		\includegraphics[height=4.6cm]{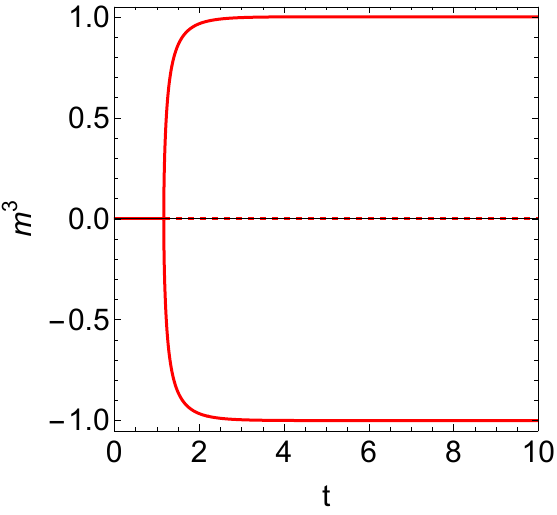}
		\caption{ Order parameters in the 2-parameter reduction for   $\lambda>1/3$. \emph{Left column}: $\lambda=0.429$. \emph{Centre column}: $\lambda=\lambda_{tc}^{(IB)}=2/3$. \emph{Right column}: $\lambda=0.857$. 
			\label{Fig:mvst large lambda}
		}
	\end{center}
\end{figure}

\section{Concluding remarks}
\label{sec:conclusion}
In this paper we have analysed in detail a discrete mean-field model for a biaxial nematic liquid crystal subject to external fields, using an approach based on the differential identity~\eqref{eq:heat} for the partition function. Upon the introduction of suitable variables, namely the order parameters, the multidimensional linear PDE satisfied by the partition function leads, in the thermodynamic limit, to a set of equations of state involving all four orientational order parameters. The equations are completely solvable by the method of characteristics proving the integrability of the model.

 Via the introduction of a novel set of order parameters corresponding to orientational Gibbs weights, we have obtained the equations of state in explicit form. We proved that, in the absence of external fields, the system is fully characterised by two-parameter reductions, and such reductions persist in the case of non-zero external fields subject to suitable constraints.  A detailed analysis demonstrates the existence of a rich phase diagram, that is remarkably consistent with the results known in the literature for the standard Maier-Saupe model and its biaxial extensions. Hence, the discrete models of the type studied in this paper capture, at least qualitatively, the most important features of continuum models with external fields for which explicit analytic formulae are not available. These results indeed encourage further studies on  integrable biaxial models where the Hamiltonian contains a more general nonlinear dependence on the tensors ${\bf q}$ and ${\bf b}$, as for instance the one implied by the full Straley pair-interaction potential. Moreover, the phenomenology encoded in the equations of state \eqref{eq:eosp 1}-\eqref{eq:eosp 4} for general field values, as well as the 2-parameter reductions obtained in Proposition~\ref{thm: reductions fields} for constrained fields, still need to be further analysed and described. Such cases are currently under investigation and results will be reported in due course.

\section*{Acknowledgements}
We would like to thank the Isaac Newton Institute for Mathematical Sciences for the hospitality during the six-month programme `Dispersive hydrodynamics: mathematics, simulation and experiments, with applications in nonlinear waves', Cambridge July-December 2022, under the EPSRC Grant Number EP/R014604/1, where this work has been partly developed, and  GNFM - Gruppo Nazionale per la Fisica Matematica, INdAM (Istituto Nazionale di Alta Matematica).  F.G.  also acknowledges the hospitality of the Department of Mathematics, Physics and Electrical Engineering of Northumbria University Newcastle.
A.M. is supported by the Leverhulme Trust Research Project Grant 2017-228, the Royal Society International Exchanges Grant IES-R2-170116 and London Mathematical Society.

\end{document}